\documentclass[11pt]{article}

\usepackage[margin=1in]{geometry}

\usepackage{amsmath,amsfonts,amsthm,latexsym,amssymb,xspace}
\usepackage{engord}
\usepackage{enumerate}
\usepackage{url}

\newtheorem{thm}{Theorem}
\newtheorem{cor}[thm]{Corollary}
\newtheorem{lem}[thm]{Lemma}

\theoremstyle{definition}

\theoremstyle{remark}

\newcommand{\set}[1]{\left\{#1\right\}}

\newcommand{\ints}{\ensuremath{\mathbb{Z}}\xspace}
\newcommand{\rats}{\ensuremath{\mathbb{Q}}\xspace}
\newcommand{\nnrats}{\ensuremath{\rats^{\geq 0}}\xspace}

\newcommand{\algs}{\ensuremath{\overline{\mathbb{Q}}}\xspace}
\newcommand{\nnalgs}{\ensuremath{\algs^{\geq 0}}\xspace}
\newcommand{\palgs}{\ensuremath{\algs^{>0}}\xspace}

\newcommand{\csp}{\ensuremath{\mathsf{CSP}}\xspace}
\newcommand{\ncsp}{\ensuremath{\mathsf{\#CSP}}\xspace}
\newcommand{\eval}{\textsf{\textup{Eval}}}

\newcommand{\nump}{\textsf{\textup{\#P}}\xspace}
\newcommand{\numpc}{\textsf{\#P}-complete\xspace}
\newcommand{\numph}{\textsf{\textup{\#P}}-hard\xspace}
\newcommand{\fp}{\textsf{\textup{FP}}\xspace}
\newcommand{\fpnump}{\ensuremath{\fp^\nump}\xspace}

\newcommand{\maltsev}{\textrm{Mal'tsev}\xspace}

\newcommand{\A}{\mathcal{A}}

\newcommand{\G}{\mathcal{G}}

\newcommand{\bA}{\mathbf{A}}

\newcommand{\halfpt}{\hspace{0.5pt}}

\let\oldmarginpar\marginpar
\renewcommand\marginpar[1]{\-\oldmarginpar[\raggedleft\footnotesize #1]%
{\raggedright\footnotesize #1}}

\def\fstar{\tilde f\halfpt}
\def\Rstar{\widetilde R}
\let\U=w
\newcommand{\uplabel}[1]{^{[#1]}}
\newcommand{\eqclass}[2]{[#1]\uplabel{#2}}
\newcommand{\eqrep}[2]{\bar #1\uplabel{#2}}

\title{A complexity dichotomy
for hypergraph partition functions\thanks{Partly funded by the EPSRC
grant ``The complexity of counting in constraint satisfaction problems''.
Some of the work was done while the authors were visiting
the ``Combinatorics and Statistical Mechanics'' programme of the Isaac
Newton Institute for Mathematical Sciences, University of Cambridge.}}
\author{Martin Dyer \\
School of Computing\\ University of Leeds\\
Leeds LS2~9JT, UK\and
Leslie Ann Goldberg \\
Department of Computer Science,\\
 University of Liverpool,\\
 Liverpool
L69 3BX, UK\and
Mark Jerrum \\
School of Mathematical Sciences,\\
Queen Mary, University of London\\
Mile End Road, London E1 4NS, UK}
\date{\today}
\begin{document}

\maketitle

\begin{abstract}
We consider the complexity of counting homomorphisms from an $r$-uniform
hypergraph $G$ to a symmetric $r$-ary relation $H$. We give a dichotomy theorem
for $r>2$, showing for which $H$ this problem is in \fp and for which $H$ it is \nump-complete. This generalises a theorem of Dyer and Greenhill (2000) for
the case $r=2$, which corresponds to counting graph homomorphisms. Our
dichotomy theorem extends to the case in which the relation $H$ is weighted,
and the goal is to compute the \emph{partition function}, which is the sum of
weights of the homomorphisms. This problem is motivated by statistical physics,
where it arises as computing the partition function for particle models in
which certain combinations of $r$ sites interact symmetrically. In the weighted
case, our dichotomy theorem generalises a result of Bulatov and Grohe (2005)
for graphs, where $r=2$. When $r=2$, the polynomial time cases of the dichotomy
correspond simply to rank-1 weights. Surprisingly, for all $r>2$ the polynomial time cases of the dichotomy have rather more structure. It turns out that the
weights must be superimposed on a combinatorial structure defined by solutions
of an equation over an Abelian group. Our result also gives a dichotomy
for a closely related constraint satisfaction problem.
\end{abstract}

\section{Introduction}
\label{sec:intro}
We consider the complexity of counting homomorphisms from an $r$-uniform
hypergraph $G$ to a symmetric $r$-ary relation $H$. We will give a
\emph{dichotomy} theorem for $r>2$, showing that counting is in 
polynomial time for certain $H$ and is \numpc for the remainder.
Moreover our dichotomy is \emph{effective}, meaning that there is an 
algorithm that takes $H$ as input and determines whether the counting problem is 
polynomial time solvable or whether it is \numpc. This generalises a theorem of 
Dyer and Greenhill~\cite{DyeGre00} for the case $r=2$, which corresponds to 
counting graph homomorphisms or $H$-colourings.

Our dichotomy extends to the case in which the relation $H$ is weighted,
and we wish to compute the \emph{partition function}, which is the sum of
weights of all homomorphisms.  Here our dichotomy theorem extends a result of
Bulatov and Grohe~\cite{BulGro05} for the case of graphs, $r=2$. In the graph
dichotomy, the polynomial time cases correspond simply to weights which
form rank-1 matrices. Surprisingly, for all $r>2$, the polynomial time solvable
cases are more structured. It turns out that the weights must be superimposed
on a combinatorial structure defined by solutions of an equation over an
Abelian group. We note that this already appears in a disguised form in
the case $r=2$. The bipartite case, which has no obvious analogue for $r>2$,
corresponds to the equation $\alpha_1+\alpha_2=1$ over the group $\ints_2$.

A motivation for considering this question comes from statistical physics.
Identifying $V(G)$ with a set of {\it sites\/} and $D$ with a set of
$q$~{\it spins},
the quantity that we wish to compute,
$Z^g(G)$, can be viewed as the partition function of
a statistical physics model in which certain sets of $r$~sites interact
symmetrically, and their interaction contributes to the \emph{Hamiltonian} of
the
system. The partition function then
gives the normalising constant for the
\emph{Gibbs distribution} of the system. The sets of $r$ interacting sites
are the edges
of~$G$.
(Sometimes, an edge of size greater than~$2$ is referred to as
a ``hyperedge'', but we do not use that terminology here.)
Clearly, the sites in an edge
should be distinct, although their
spins need not be.  In this application, the edges would usually represent
sets of sites which are in close physical proximity.

\subsection{Notation and definitions}
\label{subsec:defs}
An \emph{$r$-uniform hypergraph} $G$ was defined by Berge~\cite{Berge70} to be a
system of subsets of a set $V(G)$, where $n=|V(G)|$, in which each subset has
cardinality $r$. The elements of $V(G)$ are the \emph{vertices} of the hypergraph,
and the subsets are its \emph{edges}. Then $E(G)$ denotes the \emph{edge set} of $G$. Let $M=|E(G)|$.
Note that the edges of $G$ are distinct sets, otherwise the set system is a
\emph{multihypergraph}. Note also that the edges are sets, not \emph{multisets},
otherwise  the multiset system has been called a \emph{hypergraph with
multiplicities}~\cite{LanZie07}. Note that ``$r$-uniform hypergraph with
multiplicities'' is synonymous with ``symmetric $r$-ary relation''.  A \emph{loop} is
then a (multiset) edge in which all $r$ vertices are the same~\cite{LanZie07}.
Therefore a \emph{simple graph} $G=(V,E)$ (having no loops or \emph{parallel edges})
is a $2$-uniform hypergraph, a graph with parallel edges is a $2$-uniform
multihypergraph, and a graph with loops is a $2$-uniform hypergraph with
multiplicities, or a symmetric binary relation.

Let $D$ be a finite set with $q=|D|$.  We will assume $q\geq 2$, since
the cases $q\leq 1$ are trivial. For some $r\geq 3$, we consider a symmetric $r$-ary function $g$ with domain $D$ and codomain a set of real numbers. The codomain we will choose is the set of nonnegative algebraic numbers,  $\nnalgs$. Thus \algs denotes the field of all algebraic numbers, and we let \palgs denote the positive numbers in \algs. Our principal reason for this choice is that arithmetic operations and comparisons on such numbers can be carried out exactly on a Turing machine. See, for example,~\cite{Cohen93}. Moreover, since our analysis is entirely concerned with polynomial equations, it is natural to work in \algs, which is the algebraic closure of the rational field \rats.

Given a symmetric function $g:D^r\to \nnalgs$ and an $r$-uniform hypergraph $G$ as input, the partition function associated with $g$ is
\begin{equation}\label{eq:pf}
   Z^g(G) = \sum_{\sigma: V(G)\rightarrow D}\,
   \prod_{(u_1,\ldots,u_r)\in E(G)} g(\sigma(u_1),\ldots,\sigma(u_r)).
\end{equation}
$\eval(g)$ is the problem of computing $Z^g(G)$, given the input~$G$.
Each choice for the function~$g$ leads to a computational problem which
we will call $\eval(g)$, and we may ask how the computational
complexity of $\eval(g)$ varies with~$g$.

We may view \eqref{eq:pf} as the evaluation of a multivariate polynomial function of the \emph{weights} $g(x)$ ($x\in D^r$). If there are $N$ different irrational weights $\xi_1,\xi_2,\ldots\xi_N$, we can perform the necessary computations in the field $\rats(\xi_1,\xi_2,\ldots\xi_N)$. It is known that this field is equivalent to $\rats(\theta)$ for a single algebraic number $\theta$, the primitive element, and
an algorithm to determine $\theta$ exists. We do not need to consider the
efficiency of this algorithm, since $N$ is a constant. The standard representation of
a number in $\rats(\theta)$ is a constant degree polynomial in $\theta$ with
rational coefficients. Arithmetic operations in $\rats(\theta)$ can be carried out in
this representation.  For details, see~\cite{Cohen93}. We assume that $g$ is
pre-processed so that all weights are given in this standard representation. Some
of our intermediate reductions seemingly require computing in larger algebraic number
fields. This is true even if all original weights are rational, and justifies our choice of \algs as the codomain of $g$. We will suppose, without further comment, that the necessary algebraic numbers are adjoined to $\rats(\theta)$ as required.
In any case, we compute only in numbers fields which have constant degree
over \rats. Despite this increase in field size during our reductions, we will show that the resulting algorithm for the polynomial time solvable cases can perform its computations entirely within $\rats(\theta)$.  Note that the exact representation in $\rats(\theta)$ can also be used to compute in \fp any polynomial number of bits of the binary expansion of $Z^g(G)$, if this is required.

It is easy to bound the number of different monomials which occur in~\eqref{eq:pf}.
Suppose there are $K$ nonzero weights, for some $0\leq K \leq \binom{q}{r}$.
Then the polynomial~\eqref{eq:pf} has at most
\[\binom{M + K-1}{K-1}\ =\ O(M^{K-1})\]
monomial terms, which is polynomial in the size of the input. Each  monomial can be computed exactly in \fp, working in the field $\rats(\theta)$. The coefficient
of each monomial is an integer, which is easily seen to be computable in \nump. The nondeterministic Turing machine guesses $\sigma:V(G)\to D$, computes the term in \eqref{eq:pf} as a monomial in the weights and accepts if it is the chosen monomial. Therefore $Z^g(G)$ can be computed exactly in \fpnump as an element of $\rats(\theta)$. Consequently, showing that $Z^g(G)$ is \numph  implies that it is complete for \fpnump. We make use of this observation below.

It will be helpful to
describe a \emph{constraint satisfaction} problem
which is closely related to $\eval(g)$.
An instance~$I$ of $\ncsp(g)$ consists of
a set $V(I)=\{v_1,\ldots,v_n\}$ of \emph{variables}
and a multiset $E(I)$ of \emph{constraints}.
Each constraint has a \emph{scope}, $(u_{1},\ldots,u_{r})$,
which is a tuple of $r$ variables.
The partition function $Z^g(I)$
is given by
\begin{equation}\label{eq:pfcsp}
   Z^g(I) = \sum_{\sigma: V(I)\rightarrow D}\,
   \prod_{(u_1,\ldots,u_r)\in E(I)} g(\sigma(u_1),\ldots,\sigma(u_r)).
\end{equation}

Thus, every instance $G$ of $\eval(g)$ can be viewed as an instance of
$\ncsp(g)$
by taking the vertices as variables and the edges as constraint scopes. The
value of the partition function that gets output is the same
in both cases.
Thus, we have a trivial polynomial time reduction from
$\eval(g)$ to $\ncsp(g)$.
The opposite is not necessarily true, because a constraint scope
$(u_1,\ldots,u_r)$
of an instance $I$ of $\ncsp(g)$ might not be an edge -- the same
variable might appear more than once amongst $u_1,\ldots,u_r$.
Also, the same scope might appear more than once in~$E(I)$.
So an instance $I$ of $\ncsp(g)$ might not be a properly-formed
instance of $\eval(g)$. In fact, $I$ is a multihypergraph with
multiplicities in general, rather than a hypergraph.
Nevertheless, our main result applies also to the
problem $\ncsp(g)$ --- see Corollary~\ref{cor:free}.
We note that both the $\eval(g)$ and the $\ncsp(g)$ problems have been studied
extensively.

The problem $\ncsp(g)$ may be generalised to the case in which
the  parameter $g$ is replaced by a set of functions $\Gamma$.
If $\Gamma$ is a set of functions (of various arities) from~$D$
to \nnalgs, then $\ncsp(\Gamma)$ is the problem of computing the
partition function of an instance~$I$ in which each constraint
with $r$-ary scope specifies a particular $r$-ary function
from~$\Gamma$ which should be applied to the scope in the
partition function. See~\cite{BulGro05} or \cite{DyGoJe07} for
further details. If the functions in $\Gamma$ are not required
to have any additional properties, like symmetry or given arity,
$\ncsp(\Gamma)$ is actually no more general than $\ncsp(g)$, at
least from the viewpoint of computational complexity. It can be
shown that the two problems have the same complexity under
polynomial time reductions~\cite{Bulvisit}. Note, however, that
the reduction from $\ncsp(\Gamma)$ to $\ncsp(g)$ given
in~\cite{Bulvisit} does not preserve symmetry. So this equivalence does not
permit us to replace a family $\Gamma$ of \emph{symmetric}
functions by a single symmetric function $g$. This holds even
in the simplest possible case
in which $\Gamma$ has
two unary functions. Hence,
restricted to symmetric functions, $\ncsp(\Gamma)$ may be a more
general problem than $\ncsp(g)$, but we do not consider it
further here.

\subsection{Previous work}
\label{subsec:previous}
The computational complexity of problems of the type we consider here
was first investigated
by Dyer and Greenhill~\cite{DyeGre00}, who examined the
complexity of $\eval(g)$ in the
special case in which $r=2$ and $g:D^{2}\to\{0,1\}$, so $g$ is
equivalent to a symmetric relation on~$D$. This is the problem
of counting homomorphisms from an input simple graph~$G$ to a fixed
(undirected) graph~$H$, possibly with loops, where the function~$g$
represents the adjacency matrix of~$H$. They showed that there is a
polynomial time algorithm when each connected component of~$H$ is either
a complete unlooped bipartite graph or a complete looped graph.
In all other cases the counting problem $\eval(g)$ is \nump-complete.

More generally, Bulatov and Grohe~\cite{BulGro05}
considered the complexity of $\ncsp(g)$ when
$g$ is a symmetric binary function on $D$.
If the input is a simple graph $G$, we can think of this
as counting weighted homomorphisms from $G$ to an
undirected graph~$H$ with nonnegative edge weights.
The function $g$ is equivalent to the weighted
adjacency matrix $\bA$ of $H$.
If $H$ is connected, then we say that the matrix~$\bA$ is ``connected'', otherwise
the ``connected components'' of~$\bA$ correspond to the connected components
of the graph~$H$. Similarly, we say that $\bA$ is bipartite if and only if $H$ is bipartite.
In this setting,
Bulatov and Grohe~\cite{BulGro05} established the following
important theorem, which is central to our analysis.
\begin{thm}[Bulatov and Grohe]\label{BulGro}
Let $\bA$ be a symmetric matrix with non-negative real entries.
\begin{enumerate}
  \item[(1)] If $\bA$ is connected and not bipartite, then $\eval(\bA)$ is in polynomial time if the row rank of $\bA$ is at most 1; otherwise \eval($\bA$) is \numph.
  \item[(2)] If $\bA$ is connected and bipartite, then $\eval(\bA)$ is in polynomial time if the row rank of $\bA$ is at most 2; otherwise $\eval(\bA)$ is \numph.
  \item[(3)] If $\bA$ is not connected, then $\eval(\bA)$ is in polynomial time if each of its connected components satisfies the
corresponding condition stated in (1) or (2); otherwise $\eval(\bA)$ is \numph.
\end{enumerate}
\end{thm}

Although Theorem~\ref{BulGro} is stated for real numbers, we will make
use of it only in the case of the algebraic numbers, since it is not
clear to us how it extends to the models of real computation discussed in~\cite{BulGro05}.
We prefer to work entirely in the standard Turing machine
model of computation, though there may well be models of real computation in
which Theorem~\ref{BulGro} is valid.
For algebraic numbers, which include the
rationals, all the arithmetic operations and comparisons required in
our reductions, and those of~\cite{BulGro05}, can be carried out exactly in
the Turing machine model.

In the unweighted case of $\ncsp(\Gamma)$, where all functions in $\Gamma$ have
codomain $\{0,1\}$, Bulatov~\cite{Bulato08} has recently shown that there is a
dichotomy between those $\Gamma$ for which $\ncsp(\Gamma)$ is polynomial time
solvable, and those for which it is \nump-complete. The dichotomy can be extended to
the case in which all functions in $\Gamma$ have codomain \nnrats, the nonnegative
rational numbers, using polynomial time reductions~\cite{Bulvisit}.
However, the reductions involved do not seem to
extend to functions with codomain \nnalgs.

Establishing the existence of a dichotomy for $\ncsp(\Gamma)$ is
a major breakthrough. Nevertheless, the techniques of~\cite{Bulato08}
shed very little light on which $\Gamma$ render $\ncsp(\Gamma)$ polynomial time solvable, and which $\Gamma$ render it \numph.
In the current state
of knowledge, Bulatov's dichotomy~\cite{Bulato08} is not effective,
and its decidability is an open question.

\subsection{The new results}
Our main theorem, Theorem~\ref{thm:state}, gives a dichotomy
for the case in which $\Gamma$ contains a single
symmetric function~$g$.  For this problem, we
identify a set of functions~$g$ for which $\eval(g)$ is
computable in \fp, and we show that, for every other function~$g$,
$\eval(g)$ is complete for \fpnump.

We examine both $\eval(g)$ and $\ncsp(g)$ in this setting,
and give an explicit dichotomy theorem in both cases, extending
the theorems of Dyer and Greenhill~\cite{DyeGre00}
and Bulatov and Grohe~\cite{BulGro05} to $r>2$.
In the $r>2$ case, the problem $\eval(g)$
can be understood as evaluating sums of weighted
homomorphisms from an input hypergraph~$G$
to a fixed weighted hypergraph with multiplicities~$H$. The weights of
edges in $H$ are represented by the function $g$.

As in the $r=2$ case, there
is a dichotomy, but this time some nontrivial algebraic structure
is involved in the classification.  The polynomial time solvable
cases have rank-1 weights as before, but this time, these
weights are superimposed
on a combinatorial structure defined by solutions to an equation
over an Abelian group.
In particular, $\eval(g)$ is polynomial time solvable
if and only if
each connected piece of the domain factors
as the cartesian product of two sets $A$ and $[s]$.
Then, for any
 $\alpha_1,\ldots,\alpha_r \in A$ and
$i_1,\ldots,i_r\in [s]$,
the value of
$g( (\alpha_1,i_1),\ldots,(\alpha_r,i_r) )$
is equal to~$0$ unless
$(\alpha_1,\ldots,\alpha_r)$ is a solution to an equation in an Abelian
group with domain~$A$. In that case,
the value
$g( (\alpha_1,i_1),\ldots,(\alpha_r,i_r) )$
is just the product of some positive
weights $\lambda_{i_1},\ldots,\lambda_{i_r}$.
A ``connected piece'' of the domain is defined as follows:  two elements $z$ and $z'$
are linked if there are some $z_2, \ldots, z_{r-1}$ such that
$g(z, z_2, ..., z_{r-1}, z') >0$. In general, two elements $z$ and $z'$
are 
connected if there is a sequence of $c$ elements $z_1,\ldots,z_c$ with
$z=z_1$ and $z'=z_c$ such that each pair $(z_i,z_{i+1})$ is linked.
See Theorem~\ref{thm:state} for details.

In fact, it turns out that there is only
one way to factor the connected component of the domain into~$A$ and~$[s]$
(see Theorem~\ref{thm:main}). Thus, there is a straightforward
algorithm that takes~$g$ and determines
whether $\eval(g)$ is in \fp or is \nump-hard.
See Sections~\ref{sec:final} and~\ref{sec:modelcomp}.

Our result is in a similar spirit to
the result of Kl\'\i ma, Larose and Tesson~\cite{KlLaTe06}
which gives a dichotomy for the problem of counting the number of solutions
to a system of equations over a fixed semigroup. Although our application is
rather different, parts of our proof draw inspiration from the proof of
their theorem.

\section{The main theorem}
\label{sec:main}
For $1\leq k \leq r$, we will
define
\[f\uplabel{k}(z_1,\ldots,z_k) = \sum_{z_{k+1},\ldots,z_r\in D}
g(z_1,\ldots,z_r).\]
Note that $f\uplabel{k}$ is symmetric and that
$f\uplabel{r}(z_1,\ldots,z_r) = g(z_1,\ldots,z_r)$.
Let
\[R\uplabel{k}=\{(z_{1},\ldots,z_{k}):f\uplabel k(z_{1},\ldots,z_{k})>0\}\]
be the relation underlying~$f\uplabel{k}$.
We will view relations either as subsets of~$D^{k}$
or as functions $D^{k}\to\{0,1\}$ according to convenience.
To avoid trivialities, we assume that $R\uplabel{1}$ is the complete
relation, i.e., that all elements of $D$ participate in the
relation;  if not, an equivalent problem can be formed by simply removing the
non-participating elements from~$D$.
For any $k<r$ we have
$f\uplabel{k}(z_1,\ldots,z_k) =
\sum_{z_{k+1}\in D}f\uplabel{k+1}(z_1,\ldots,z_{k+1})$
so if $k\geq 2$ then
$(z_1,z_2)\in R\uplabel{2}$ is equivalent to
``there exist $z_3,\ldots,z_k$ such that $(z_1,\ldots,z_k)\in R\uplabel{k}$''.
Let $\equiv$ be the equivalence relation which is the
transitive, reflexive closure of~$R\uplabel{2}$.
The domain $D$ is partitioned into
equivalence classes (``connected components'')
$D=D_{1}\cup\cdots\cup D_{m}$ by~$\equiv$.

We will use the following notation: We will
let $\ell$ range over $[m]$, and use it
to refer to a particular connected component~$D_\ell$.  When
applied to any function
as a subscript, it denotes the
restriction of that function to the relevant connected component.
For example,
 $f\uplabel{k}_{\ell}:(D_{\ell})^{k}\to\nnalgs$ denotes the
restriction of $f\uplabel{k}$ to the $\ell$th  connected
component~$D_{\ell}$. Likewise, $g_{\ell}$ is the restriction of~$g$
to~$D_{\ell}$.
Given the definition of $\equiv$, it is clear that
 $f\uplabel{k}=f\uplabel{k}_{1}\oplus\cdots\oplus f\uplabel{k}_{m}$ (meaning
 that
$f\uplabel{k}(z_1,\ldots,z_k)=0$ unless $z_1,\ldots,z_k$ are all in the same
connected component).
We can now state the main theorem.

\begin{thm}
\label{thm:state}
Let $g:D^r \rightarrow \nnalgs$ be
a symmetric function
with arity $r\geq 3$ and connected components $D_{1},\ldots,D_{m}$
as above.  If $g$ satisfies the following conditions,
for all $\ell\in[m]$, then $\eval(g)$
is in \fp. Otherwise, $\eval(g)$ is complete for \fpnump.
Moreover, the dichotomy is effective.
\begin{itemize}
\item There is a
set $A_\ell$ and a
positive integer $s_\ell$,
such that $D_{\ell}$ is the
Cartesian product of $A_\ell$ and $[s_\ell]$ (which we
write as
$D_{\ell}\cong A_{\ell}\times[s_{\ell}]$).
\item There are positive constants
$\{\lambda_{\ell,i}:i\in[s_{\ell}]\}$ and a relation
$S_{\ell}\subseteq A_{\ell}$ such that,
for $\alpha_1,\ldots,\alpha_r \in A_\ell$ and
$i_1,\ldots,i_r\in [s_\ell]$,
\begin{equation*}
g_{\ell}( (\alpha_1,i_1),\ldots,(\alpha_r,i_r) ) =
   \lambda_{\ell,i_1}\cdots\lambda_{\ell,i_r}
S_\ell(\alpha_1,\ldots,\alpha_r).
\end{equation*}
\item There is an Abelian group $(A_\ell,+)$
and an equation $\alpha_1 + \cdots +\alpha_r = a$ (for some element
$a\in A_\ell$) which defines $S_\ell$ in the
sense that
$(\alpha_1,\ldots,\alpha_r) \in S_\ell$
if and only if $\alpha_1 + \cdots +\alpha_r = a$.
\end{itemize}
\end{thm}
The algorithm used in the polynomial time solvable cases
of Theorem~\ref{thm:state} still works if the instance is
a \csp instance rather than a hypergraph.
Thus, we have the following corollary.

\begin{cor}
\label{cor:free}
Let $g:D^r \rightarrow \nnalgs$ be
a symmetric function
with arity $r\geq 3$ and connected components $D_{1},\ldots,D_{m}$
as above.  If $g$ satisfies the conditions in Theorem~\ref{thm:state}
for all $\ell\in[m]$, then $\ncsp(g)$
is in \fp. Otherwise,
$\ncsp(g)$ is
complete for \fpnump.
Moreover, the dichotomy is effective.
\end{cor}

Some of the \nump-hardness proofs in the proof of Theorem~\ref{thm:state}
could be simplified if we allowed ourselves a general CSP instance
rather than a hypergraph, but we   refrain from using this simplification
in order to obtain the strongest-possible result
(that is, to obtain Theorem~\ref{thm:state} rather than just
Corollary~\ref{cor:free}).

\section{A restatement of the main theorem}
We introduce some further notation and restate the main theorem
more compactly.  Along the way we gather more information,
e.g., about the factorization $D_{\ell}\cong A_{\ell}\times[s_{\ell}]$.

We define the equivalence relation $\sim_k$ on $D$
as follows: $z_1 \sim_k z'_1$ iff
there is a $\lambda$ in $\palgs$ such that, for
all $z_2,\ldots,z_k\in D$,
$f\uplabel{k}(z_1,z_2,\ldots,z_k) = \lambda
f\uplabel{k}(z'_1,z_2,\ldots,z_k)$.
Note that $\sim_k$ refines $\sim_{k-1}$. Also, $\sim_{2}$
refines~$\equiv$
since,
for any $z_{1},z_{1}'\in D$, $z_{1}\sim_{2}z'_{1}$ implies that
there exists $z_{2}$ satisfying $R\uplabel2(z_{1},z_{2})$ and
$R\uplabel2(z_{1}',z_{2})$, which in turn implies $z_{1}\equiv z_{1}'$.

Let $\eqclass{x}{k}=\set{y:y\sim_k x}$ be the equivalence class of~$x$
under~$\sim_k$. Choose a unique representative
$\eqrep{x}{k}\in \eqclass{x}{k}$.
Thus $\eqrep{x}{k}=\eqrep{y}{k}$ if and only if
$x\sim_k y$.
Let $A\uplabel{k}=\set{\eqrep{x}{k} :
x\in D}$.
Let $A\uplabel{k}_\ell$ denote the restriction of~$A\uplabel{k}$ to $D_\ell$
so $A\uplabel{k}_\ell=\set{\eqrep{x}{k}:
x\in D_\ell}$.

Note that $R\uplabel{k}$ is consistent with $\sim_k$
in the sense that $R\uplabel{k}(z_1,\ldots,z_k)=R\uplabel{k}(
\eqrep{z_1}{k},\ldots,\eqrep{z_k}{k})$,
so we can
quotient $R\uplabel{k}$ by~$\sim_k$ to get
a relation $S\uplabel{k}=R\uplabel{k}/{\sim_k}$ on~$A\uplabel{k}$.
Note that $S\uplabel{k}$ is
just the restriction of~$R\uplabel{k}$ to~$A\uplabel{k}$.
Also, $S\uplabel{k}_\ell$ is the restriction
of $R_\ell\uplabel{k}$ to $A_\ell\uplabel{k}$.

Suppose $k$ is in the range $2\leq k \leq r$.
We say that $g$ is
\emph{$k$-factoring} if the following conditions hold for every $\ell\in[m]$.
\begin{enumerate}
\item There is a positive integer $s\uplabel{k}_\ell$
such that $D_{\ell}$ is the
Cartesian product of $A\uplabel{k}_\ell$ and $[s\uplabel{k}_\ell]$ (which
we
write as
$D_{\ell}\cong A\uplabel{k}_{\ell}\times[s\uplabel{k}_{l}]$).
\item There are positive constants
$\{\lambda\uplabel{k}_{\ell,i}:
i\in[s\uplabel{k}_{\ell}]\}$ such that,
for $\alpha_1,\ldots,\alpha_k \in A_\ell\uplabel k$ and
$i_1,\ldots,i_k\in [s_\ell\uplabel k]$,
\[
f\uplabel{k}_{\ell}( (\alpha_1,i_1),\ldots,(\alpha_k,i_k) ) =
   \lambda\uplabel{k}_{\ell,i_1}\cdots\lambda\uplabel{k}_{\ell,i_k}
S\uplabel{k}_\ell(\alpha_1,\ldots,\alpha_k).
\]
\end{enumerate}
If $g$ is $k$-factoring then we say that $g$ is \emph{$k$-equational}
if, for every  $\ell\in[m]$,
there is an Abelian group $(A\uplabel{k}_\ell,+)$
and an equation $\alpha_1 + \cdots +\alpha_k = a$ (for some element
$a\in A\uplabel{k}_\ell$) which defines $S\uplabel{k}_\ell$ in the
sense that
$(\alpha_1,\ldots,\alpha_k) \in S\uplabel{k}_\ell$
if and only if $\alpha_1 + \cdots +\alpha_k = a$.

Our main theorem (Theorem~\ref{thm:state}) can be restated as follows:
\begin{thm}
\label{thm:main}
Let $g:D^r \rightarrow \nnalgs$ be
a symmetric function
with arity $r\geq 3$.
If  $g$ is $r$-factoring and $r$-equational
then $\eval(g)$
is in \fp.  Otherwise, $\eval(g)$ is complete for \fpnump.
Moreover, the dichotomy is effective.
\end{thm}

Before proving Theorem~\ref{thm:main},
we prove that it is equivalent to Theorem~\ref{thm:state}.
First, it is easy to see that if $g$ satisfies the conditions of
Theorem~\ref{thm:main} (that is, it is $r$-factoring and $r$-equational)
then it also satisfies the conditions of Theorem~\ref{thm:state}
(taking $A_\ell$ to be $A_\ell\uplabel r$, $s_\ell$ to be
$s_\ell \uplabel r$, and
$\lambda_{\ell,i}$ to be $\lambda_{\ell,i}\uplabel r$).
The other direction is a little less obvious. Suppose that $g$ satisfies
the conditions of Theorem~\ref{thm:state}.
Fix any $\ell\in [m]$.
From the first condition of Theorem~\ref{thm:state}, we have $D_{\ell}\cong
A_{\ell}\times[s_{\ell}]$.
Consider any $\alpha,\alpha'\in A_\ell$ and
any $i,i'\in [s_\ell]$. We will argue that $(\alpha,i)\sim_r (\alpha',i')$
if and only if $\alpha=\alpha'$.
First, suppose $\alpha=\alpha'$. Then,
for any $\alpha_2,\ldots,\alpha_r\in A_\ell$ and
$i_2,\ldots,i_r\in [s_\ell]$, the second condition of Theorem~\ref{thm:state}
gives
$$g_{\ell}( (\alpha,i),(\alpha_2,i_2),\ldots,(\alpha_r,i_r) ) =
   \lambda_{\ell,i} \lambda_{\ell,i_2}\cdots\lambda_{\ell,i_r}
S_\ell(\alpha,\alpha_2,\ldots,\alpha_r)$$
and
$$g_{\ell}( (\alpha',i'),(\alpha_2,i_2),\ldots,(\alpha_r,i_r) ) =
   \lambda_{\ell,i'} \lambda_{\ell,i_2}\cdots\lambda_{\ell,i_r}
S_\ell(\alpha,\alpha_2,\ldots,\alpha_r),$$
so, by the definition of $\sim_r$,
$(\alpha,i)\sim_r (\alpha',i')$. Next, suppose $(\alpha,i)\sim_r
(\alpha',i')$.
Then there is a positive constant $\lambda$ such
that, for any $\alpha_2,\ldots,\alpha_r\in A_\ell$ and
$i_2,\ldots,i_r\in [s_\ell]$,
$$\lambda_{\ell,i} \lambda_{\ell,i_2}\cdots\lambda_{\ell,i_r}
S_\ell(\alpha,\alpha_2,\ldots,\alpha_r) = \lambda
\lambda_{\ell,i'} \lambda_{\ell,i_2}\cdots\lambda_{\ell,i_r}
S_\ell(\alpha',\alpha_2,\ldots,\alpha_r).$$
We conclude that, for any $\alpha_2,\ldots,\alpha_r\in A_\ell$,
$S_\ell(\alpha,\alpha_2,\ldots,\alpha_r) =
S_\ell(\alpha',\alpha_2,\ldots,\alpha_r)$.
By the third condition in Theorem~\ref{thm:state}, we conclude that
$\alpha=\alpha'$.
We have now shown that $(\alpha,i)\sim_r (\alpha',i')$
if and only if $\alpha=\alpha'$. This implies that we
can take the set $A_\ell\uplabel r$
of unique representatives
to be $A_\ell$ and we can take $s_\ell\uplabel r$ to be $s_\ell$.
Then, taking $\lambda_{\ell,i}\uplabel r$ to be $\lambda_{\ell,i}$,
$g$ is $r$-factoring and $r$-equational (so it satisfies
the conditions of Theorem~\ref{thm:main}).
So we conclude that the two theorems are equivalent.

Now that we have shown that Theorem~\ref{thm:main} is equivalent to
Theorem~\ref{thm:state},
the rest of the paper will focus on proving Theorem~\ref{thm:main}.
The case $r=2$ is that of weighted {\it graph\/} homomorphism,
which was analysed by Bulatov and Grohe~\cite{BulGro05}.
Theorem~\ref{thm:main} is true also when $r=2$.
In this situation, it could be viewed
as a restatement of their result.  Note, however,
that ``$2$-equational'' is a restricted notion
that places severe constraints on the groups
$(A\uplabel{2}_\ell,+)$ that can arise.
Indeed the only possibilities that are consistent with
the connectivity relation $\equiv$ are the 2-element
group $C_{2}$ (``bipartite component'') and
the trivial group (``non-bipartite component'').

It will follow from the proof of Theorem~\ref{thm:main} (assuming
that $\nump\not\subseteq \mathrm{FP}$)
that a symmetric function $g$ of arity~$r\geq3$
that is $r$-factoring and $r$-equational is
$k$-factoring and $k$-equational for all $2\leq k<r$.
In fact, the Abelian groups $(A\uplabel{k}_{\ell},{+})$
will all be trivial for $k<r$:  non-trivial group structure
is only possible at the top level.  As a first step in the
proof of Theorem~\ref{thm:main}, we verify that non-trivial group structure
is only possible at the top level.

\begin{lem}
\label{lem:simple}
Let $g:D^r \rightarrow \nnalgs$ be
a symmetric function
with arity $r\geq 3$.
If $g$ is $k$-factoring and $k$-equational for some $k<r$
then
for every $\ell\in [m]$
there are positive constants
$\{\lambda\uplabel{k}_{\ell,i}:
i\in D_{\ell}\}$ such that,
for $i_1,\ldots,i_k\in D_{\ell}$,
\[
f\uplabel{k}_{\ell}( i_1,\ldots,i_k ) =
   \lambda\uplabel{k}_{\ell,i_1}\cdots\lambda\uplabel{k}_{\ell,i_k}
.\]
\end{lem}

\begin{proof}
$g$ is $k$-factoring so
$D_{\ell}\cong A\uplabel{k}_{\ell}\times[s\uplabel{k}_{l}]$ and,
for $\alpha_1,\ldots,\alpha_k \in A_\ell\uplabel k$ and
$i_1,\ldots,i_k\in [s_\ell\uplabel k]$,
\[
f\uplabel{k}_{\ell}( (\alpha_1,i_1),\ldots,(\alpha_k,i_k) ) =
   \lambda\uplabel{k}_{\ell,i_1}\cdots\lambda\uplabel{k}_{\ell,i_k}
S\uplabel{k}_\ell(\alpha_1,\ldots,\alpha_k).
\]
Now consider
$\alpha_1,\ldots,\alpha_{k+1} \in A_\ell\uplabel k$ and
$i_1,\ldots,i_{k+1}\in [s_\ell\uplabel k]$.
If \[f\uplabel{k+1}_{\ell}( (\alpha_1,i_1),\ldots,(\alpha_{k+1},i_{k+1}) )>0\]
then
\[f\uplabel{k}_{\ell}(
(\alpha_1,i_1),\ldots,(\alpha_{k-1},i_{k-1}),(\alpha_{k},i_{k}) )>0\]
and
\[f\uplabel{k}_{\ell}(
(\alpha_1,i_1),\ldots,(\alpha_{k-1},i_{k-1}),(\alpha_{k+1},i_{k+1}) )>0.\]
So  since $g$ is $k$-equational,
\[\alpha_1+ \cdots + \alpha_{k-1} + \alpha_k =
\alpha_1 + \cdots + \alpha_{k-1} + \alpha_{k+1} = a\] so
$\alpha_k=\alpha_{k+1}$.
By symmetry, $\alpha_1 = \cdots = \alpha_{k+1}$.

%
Now if $((\alpha,i),(\beta,j))\in R\uplabel 2$ 
then there exist $\alpha_2, \ldots, \alpha_k$ 
and $i_2, \ldots, i_{k}$  such that
$$f\uplabel{k+1} ( (\alpha, i), (\beta, j), (\alpha_2, i_2), \ldots, 
(\alpha_k, i_{k}) ) >0.$$ 
Thus, $\alpha=\beta$.
Taking the transitive closure,
we note that if
$(\alpha,i)$ and $(\beta,j)$ are both in $D_\ell$
then
 $\alpha=\beta$.
Hence $|A_\ell\uplabel k|=1$ so $D_\ell = [s_\ell \uplabel k]$.
\end{proof}

Our strategy for proving Theorem~\ref{thm:main}
is  now as follows.  Suppose $\eval(g)$ is {\it not\/}
\nump-hard.  We prove, for $k=2,3,\ldots,r$ in turn, that $g$
is $k$-factoring and $k$-equational.  For $k=2$ this follows
straightforwardly from Theorem~\ref{BulGro}.
The inductive step from $k$ to $k+1$ is where the work lies,
but Lemma~\ref{lem:simple} plays a role.  Ultimately, we deduce
that $g$ is $r$-factoring and $r$-equational.  Conversely,
if $g$ is $r$-factoring and $r$-equational, the partition function
$Z^{g}$ may be computed in polynomial time
using existing algorithms for counting
solutions to systems over Abelian groups, and hence $\eval(g)$
is polynomial time solvable.

\section{Preliminaries}
\label{sec:gadget}
An easy observation that will be frequently used in the
rest of this paper is the following.
\begin{lem}\label{lem:veasy}
If $\eval(f\uplabel k)$ is \nump-hard, for some $2\leq k<r$, then
so is~$\eval(g)$.
\end{lem}
\begin{proof}
An instance of $\eval(f\uplabel k)$ is a $k$-uniform
hypergraph.  Simply pad each edge $e=(u_{1},\ldots,u_{k})$
to size~$r$ by adding $r-k$ fresh vertices as follows:
$(u_{1},\ldots,u_{k},z^{e}_{k+1},\ldots,z^{e}_{r})$.
It is easy to verify that this is a polynomial time
reduction from $\eval(f\uplabel k)$ to $\eval(g)$.
\end{proof}
Another easy observation is that the partition function
$Z^{g}(G)$ factorises if $G$ is not connected.
So we may assume henceforth that the instance
hypergraph~$G$ is connected.

For $z\in D$
let ${\lambda'}\uplabel{k}_{z}$ be defined so that,
for all $z_2,\ldots,z_k\in D$,
\[f\uplabel{k}(z,z_2,\ldots,z_k)
={\lambda'}_{z}\uplabel{k} f\uplabel{k}(
\eqrep{z}{k},z_2,\ldots, z_k).\]
(Recall from the definition of $\sim_k$ that 
${\lambda'}_z\uplabel{k}$ does not depend on 
$z_2, ..., z_k$.)
Then, by symmetry, we have
\begin{equation}\label{eq:lambdaprime}
f\uplabel{k}(z_1,\ldots,z_k) =
{\lambda'}_{z_1}\uplabel{k}
\cdots
{\lambda'}_{z_k}\uplabel{k}
f\uplabel{k} (\eqrep{z_1}{k},\ldots,\eqrep{z_k}{k}).
\end{equation}

Define
$$\fstar\uplabel{k}(z_1,z'_1) = \sum_{z_{2},\ldots,z_k\in D}
f\uplabel{k}(z_1,z_2,\ldots,z_k)
f\uplabel{k}(z'_1,z_2,\ldots,z_k).$$
Let $\Rstar\uplabel{k}$
be the (symmetric)
binary relation underlying $\fstar\uplabel{k}$.
It will turn out that $\Rstar\uplabel{k}$ and $\sim_{k}$ coincide
when $g$ is not \nump-hard.

For the purposes of this paper, a symmetric relation $R
\subset \mathcal{A}^k$
is said to be a {\it Latin hypercube\/} if, for all
$\alpha_1,\ldots,\alpha_{k-1}\in \mathcal{A}$,
there exists a unique
$\alpha_k\in\mathcal{A}$ such that $(\alpha_1,\ldots,\alpha_k)\in R$.
Note that symmetry implies similar statements with the $\alpha_i$s permuted.
This definition specialises to the familiar notion of Latin
square
if we take $k=3$ and
think of $\alpha_1$, $\alpha_2$ and $\alpha_3$
as ranging over rows,
columns and symbols, respectively.
For $k>3$ it is consistent with the existing, if less familiar,
notion of Latin $(k-1)$-hypercube.

We use the following
interpolation result, which is \cite[Lemma 3.2]{DyeGre00}
\begin{lem}
\label{lem:interp}
Let $\eta_1,\ldots,\eta_m$ be known distinct nonzero constants
Suppose that
we know values $Z_1,\ldots,Z_m$ such that
$Z_p = \sum_{\ell=1}^m \gamma_\ell \eta_\ell^p$ for
$1\leq p \leq m$. The coefficients
$\gamma_1,\ldots,\gamma_m$ can be evaluated in polynomial time.
\end{lem}

Lemma~\ref{lem:interp} has the following consequence,
since if we have $\eta_i=\eta_j$ below
we can combine $\gamma_i$ and $\gamma_j$ into $\gamma_i+\gamma_j$.

\begin{cor}
\label{cor:interp}
Let $\eta_1,\ldots,\eta_m$ be known nonzero constants
Suppose that
we know values $Z_1,\ldots,Z_m$ such that
$Z_p = \sum_{\ell=1}^m \gamma_\ell \eta_\ell^p$ for $1\leq p \leq m$. The
value $Z_0 = \sum_{\ell=1}^m \gamma_\ell$ can be computed
in polynomial time.
\end{cor}

As mentioned earlier, the base case ($k=2$)
in the proof of Theorem~\ref{thm:main} will follow from the
result of Bulatov and Grohe~\cite{BulGro05}.
They examined the complexity of $\ncsp(g)$   and
there is no  immediate polynomial time reduction from
$\ncsp(g)$ to $\eval(g)$.
The next lemma
provides such a reduction for the case that we require.
\begin{lem}
\label{lem:twostretch}
Suppose $h:D^{2}\to\nnalgs$ has connected components
$D_{1},\ldots,D_{\ell}$, and underlying relation~$R_{h}$.
Suppose also that $R_{h}$ has no bipartite components.
If the restriction $h_{\ell}$ of~$h$ to any component~$D_{\ell}$
is not rank~1, then $\eval(h)$ is \nump-hard.
\end{lem}
\begin{proof}
Let $I$ be an instance of $\ncsp(h)$.  View $I$ as a multigraph
with possible loops and parallel edges.
Form the graph~$G$ as the ``2-stretch''
of~$I$;  that is to say, subdivide each edge of~$I$ by introducing
a new vertex.  Note that $G$ is a simple graph without loops.
Define the symmetric function
$h^{(2)}:D^{2}\to\nnalgs$ by
$h^{(2)}(x,y)=\sum_{z\in D}h(x,z)h(y,z)$. Note that
$Z^{h}(G)=Z^{h^{(2)}}(I)$, and hence $\ncsp(h^{(2)})$ reduces to
$\eval(h)$.

Suppose $\eval(h)$ is not \nump-hard.
Then $\ncsp(h^{(2)})$ is not \nump-hard.
By~\cite[Thm~1(1)]{BulGro05},
$h^{(2)}$, viewed as a matrix, is a direct sum of
rank-1 matrices; i.e., each $h^{(2)}_{\ell}$ has rank~1.
But each $h^{(2)}_{\ell}$ is the ``Gram matrix'' of~$h_{\ell}$
(the product of $h_\ell$, viewed as a matrix, and its transpose),
and it is a elementary fact that the rank of a matrix and
its corresponding Gram matrix are equal~\cite{Mirsky90}.
Thus, for all~$\ell$, the restrictions~$h_{\ell}$ of~$h$
to~$D_{\ell}$ are rank~1.
\end{proof}

\section{Factoring}
\begin{lem}
\label{lem:base}
Let $g:D^r \rightarrow \nnalgs$ be
a symmetric function
with arity $r\geq 3$.
Either\/ $\eval(f\uplabel 2)$ is \nump-hard
(which implies that $\eval(g)$ is \nump-hard)
or $g$ is $2$-factoring and $2$-equational.
\end{lem}

\begin{proof}
First, note that $R\uplabel 2$ has no bipartite components:
If $(z_1,z_2)\in R\uplabel 2$
then there is a $z_3$ such that
$(z_1,z_2,z_3) \in R\uplabel 3$. By the symmetry of $f\uplabel 3$,
we find that $(z_1,z_3)$ and $(z_2,z_3)$ are also in $R\uplabel 2$,
so the component containing $z_1$ and $z_2$
is not bipartite.

Now,
by \cite{BulGro05} (using Lemma~\ref{lem:twostretch}),
$f\uplabel2_\ell$ has rank~1.
Thus, there are positive constants
$\{ \mu_z:
z\in D\}$  such
that, for every $\ell\in[m]$ and every
$z_1,z_2$ in $D_\ell$, the following holds.
  \begin{equation}\label{eq:rank1}
    f\uplabel 2_{\ell}(z_1,z_2)=\mu_{z_1}\mu_{z_2}.
\end{equation}
We conclude that
all elements in $D_\ell$ are related by $\sim_2$, so $|A\uplabel 2_\ell|=1$.
Thus, we can take $s\uplabel 2_\ell=|D_\ell|$
and $\lambda\uplabel 2 _{\ell,z} = \mu_z$
and the trivial equation (since $|A\uplabel 2_\ell|=1$).

The parenthetical claim in the statement of this lemma
and subsequent ones comes from Lemma~\ref{lem:veasy}.
\end{proof}

\begin{lem}
\label{lem:one}
Let $g:D^r \rightarrow \nnalgs$ be
a symmetric function
with arity $r\geq 3$.
Let $k$ be an integer in $\{3,\ldots,r\}$. Suppose
that $g$ is $(k-1)$-factoring and $(k-1)$-equational.
Either\/ $\eval(f\uplabel k)$ is \nump-hard
(which implies that\/ $\eval(g)$ is \nump-hard), or all the following
hold:
(i)~there are positive constants
$\{\lambda\uplabel k_{z}:
z\in D\}$
such that
$f\uplabel k(z_1,\ldots,z_k)
=\lambda\uplabel k_{z_1}\cdots\lambda\uplabel k_{z_k}\,R\uplabel k
(z_1,\ldots,z_k)$,
(ii)~for every connected component $\ell\in[m]$, the relation
$S\uplabel k_\ell$
is a Latin hypercube, and
(iii)~for every $\ell\in[m]$,
the sum
$
 \sum_{z\in\eqclass{\alpha}{k}}\lambda\uplabel k_{z}$
is
independent of $\alpha\in A\uplabel k_{\ell}$.
\end{lem}

\begin{proof}
Assume $\eval(f\uplabel k)$ is not \numph.
Fix $\ell\in[m]$ and $z_1,z'_1\in D_\ell$.
By the Cauchy-Schwarz inequality,
\begin{align*}
& \Big(\sum_{z_2,\ldots,z_k\in D_{\ell}}f_{\ell}\uplabel k
(z_1,z_2,\ldots,z_k)f_{\ell}\uplabel k(z'_1,z_2,\ldots,z_k)\Big)^2
\leq \\
& \quad \sum_{z_2,\ldots,z_k\in D_{\ell}}f_{\ell}\uplabel
k(z_1,z_2,\ldots,z_k)^2
     \sum_{z_2,\ldots,z_k\in D_{\ell}}f_{\ell}\uplabel
     k(z'_1,z_2,\ldots,z_k)^2,
\end{align*}
i.e.,
\begin{equation}\label{eq:30}
\fstar\uplabel k_{\ell}(z_1,z'_1)^{2}\leq \fstar\uplabel k_{\ell}(z_1,z_1)
\fstar_{\ell}\uplabel k(z'_1,z'_1),
\end{equation}
with equality precisely when   $z_1 \sim_k z'_1$.
Note that the difference between the right-hand-side
and the left-hand-side in Equation~(\ref{eq:30}) can be seen as a $2$~by~$2$ determinant.

Now  $\eval(\fstar\uplabel k) \leq \eval(f\uplabel k)$
since $\fstar\uplabel k(u,v)$ can be simulated by a pair of
constraints
\[f\uplabel k(u,w_2,\ldots,w_k) f\uplabel k(v,w_2,\ldots,w_k)\]
using new variables $w_2,\ldots,w_k$,
so
$\eval(\fstar\uplabel k)$ is not
\nump-hard.
$\Rstar\uplabel k$ has no bipartite components since it is reflexive,
so by \cite{BulGro05} and Lemma~\ref{lem:twostretch},
$\fstar\uplabel k$ decomposes into
a sum of rank-1 blocks.

When
$z_1 \not\sim_k z'_1$
we have strict inequality
in~(\ref{eq:30}), which implies
\begin{equation}
    \label{eq:disjoint}
\fstar\uplabel k_{\ell}(z_1,z'_1) =
\sum_{z_{2},\ldots,z_k\in D}
f\uplabel{k}(z_1,z_2,\ldots,z_k)
f\uplabel{k}(z'_1,z_2,\ldots,z_k)=0,
\end{equation}
since otherwise $\fstar\uplabel k$ would not decompose into rank~1 blocks.

So for each choice
of canonical representatives
$\alpha_2,\ldots,\alpha_k$
in~$A_\ell\uplabel k$ there is
at most one representative $\alpha_1\in A_\ell\uplabel k$
such that ${f_\ell}\uplabel k(\alpha_1,\ldots,\alpha_k)>0$.
There is at least one such representative $\alpha_1$
since, by Lemma~\ref{lem:simple},
$$
f\uplabel{k-1}_{\ell}( \alpha_2,\ldots,\alpha_k ) =
   \lambda\uplabel{k-1}_{\ell,\alpha_2}\cdots\lambda\uplabel{k-1}_{\ell,\alpha_k}
,$$ and the $\lambda\uplabel{k-1}_{\ell,\alpha_j}$ values are positive.
This is part~(ii) of the lemma.

Recall the definition of ${\lambda'}_{z}\uplabel{k}$
from Equation~(\ref{eq:lambdaprime}).
For $\alpha\in A_\ell\uplabel k$,
let
$\bar\lambda_{\alpha}$ denote the sum $\bar\lambda_{\alpha}=
\sum_{z\in\eqclass{\alpha}{k}}
{\lambda'}_z\uplabel{k}
$. Similarly, let
$\bar\mu_{\alpha}=
\sum_{z\in\eqclass{\alpha}{k}}
{\lambda_{\ell,z}}\uplabel{k-1}$.
Fix $z_2,\ldots,z_k\in D_\ell$.
By Lemma~\ref{lem:simple},
\begin{align*}
 \lambda\uplabel{k-1}_{\ell,z_2}\cdots\lambda\uplabel{k-1}_{\ell,z_{k}}
= f_\ell\uplabel{k-1}(z_2,\ldots,z_k) & =
\sum_{z_1\in D_\ell} f_\ell \uplabel k(z_1,\ldots,z_k)\\
& =
\sum_{z_1\in D_\ell}
{\lambda'}_{z_1}\uplabel{k}
\cdots
{\lambda'}_{z_k}\uplabel{k}\,
f_\ell\uplabel{k} (\eqrep{z_1}{k},\ldots,\eqrep{z_k}{k})\\
& =
\bar\lambda_{\alpha_1}
{\lambda'}_{z_2}\uplabel{k}
\cdots
{\lambda'}_{z_k}\uplabel{k}\,
f_\ell\uplabel{k} (\alpha_1,\eqrep{z_2}{k},
\ldots,\eqrep{z_k}{k}),
\end{align*}
where $\alpha_1$ is the unique representative in $A_\ell\uplabel k$
such that $f\uplabel{k} (\alpha_1,\eqrep{z_2}{k},\ldots,\eqrep{z_k}{k})>0$.
So for fixed $\alpha_2,\ldots,\alpha_k\in A_\ell\uplabel k$,
there is a representative $\alpha_1 \in A_\ell \uplabel k$ such that
\begin{align*}
\bar\mu_{\alpha_2}\cdots\bar\mu_{\alpha_k} & =
\sum_{z_2\in\eqclass{\alpha_2}{k}}\cdots
\sum_{z_k\in\eqclass{\alpha_k}{k}}
 \lambda\uplabel{k-1}_{\ell,z_2}\cdots\lambda\uplabel{k-1}_{\ell,z_{k}}\\
& =
\sum_{z_2\in\eqclass{\alpha_2}{k}}\cdots
\sum_{z_k\in\eqclass{\alpha_k}{k}}
\bar\lambda_{\alpha_1}
{\lambda'}_{z_2}\uplabel{k}
\cdots
{\lambda'}_{z_k}\uplabel{k}
f_\ell\uplabel{k} (\alpha_1,
\ldots,\alpha_k)
\\
& =
\bar\lambda_{\alpha_1}
\cdots
\bar\lambda_{\alpha_k}
f_\ell\uplabel{k} (\alpha_1,
\ldots,\alpha_k).
\end{align*}

Since we have a Latin hypercube (Part (ii) of the lemma), any of
$\alpha_1,\ldots,\alpha_k$ is determined by the other $k-1$ of them.
Thus, we can derive a similar equality omitting any other $ \bar\mu_{\alpha_i}$
on the left-hand-side.
Now the right-hand-side of the above equality is symmetric
in the $\alpha_j$'s, and the
left-hand-side has exactly one $\alpha_j$ missing, so
by symmetry
we conclude $\bar\mu_{\alpha_1}=\cdots=\bar\mu_{\alpha_k}$
and, further, $\bar\mu_{\alpha_j}$ is constant for
$\alpha_j \in A_\ell\uplabel k$.
Moreover,
$\bar\lambda_{\alpha_1}
\cdots
\bar\lambda_{\alpha_k}\,
f_\ell\uplabel{k} (\alpha_1,
\ldots,\alpha_k)$
is constant on representatives
$\alpha_1,\ldots,\alpha_k \in A_\ell\uplabel k$
with
$f_\ell\uplabel{k} (\alpha_1,
\ldots,\alpha_k)>0$.
That is, for any set of $k$~representatives
$\alpha'_1, \alpha'_2, \ldots, \alpha'_k\in A_\ell\uplabel k$
with $f_\ell\uplabel{k} (\alpha'_1,
\ldots,\alpha'_k)>0$,
the value of that expression 
$\bar\lambda_{\alpha'_1}
\cdots
\bar\lambda_{\alpha'_k}
f_\ell\uplabel{k} (\alpha'_1,
\ldots,\alpha'_k)$
is the same.

Now define $\lambda_x\uplabel k =
 c_{\ell} {\lambda'}_{x}\uplabel k /\bar\lambda_{{\eqclass{x}{k}}}$,
where $c_{\ell}$ is a constant, depending only on~$\ell$,
to be determined below.
Then, whenever
$f_\ell\uplabel{k}(z_1,\ldots,z_k)>0$,
\begin{align*}
f_{\ell}\uplabel{k}(z_1,\ldots,z_k) &=
{\lambda'}_{z_1}\uplabel{k}
\cdots
{\lambda'}_{z_k}\uplabel{k}\,
f_\ell\uplabel{k} (\eqrep{z_1}{k},\ldots,\eqrep{z_k}{k})\\
    & =c_{\ell}^{-k}
\lambda_{z_1}\uplabel{k}
\cdots
\lambda_{z_k}\uplabel{k}
\bar\lambda_{\eqclass{z_1}{k}}
\cdots
\bar\lambda_{\eqclass{z_k}{k}}\,
f_{\ell}\uplabel{k}(\eqrep{z_1}{k},\ldots,\eqrep{z_k}{k}).
\end{align*}

But
$$c_{\ell}^{-k}
\bar\lambda_{\eqclass{z_1}{k}}
\cdots
\bar\lambda_{\eqclass{z_k}{k}}
f_{\ell}\uplabel{k}(\eqrep{z_1}{k},\ldots,\eqrep{z_k}{k})$$
is independent of $z_1,\ldots,z_k$ (assuming, as we are, that
$f_\ell\uplabel k(z_1,\ldots,z_k)>0$),
so, by appropriate choice of~$c_{\ell}$,
$$
f_{\ell}\uplabel{k}(z_1,\ldots,z_k) =
{\lambda_{z_1}}\uplabel{k}
\cdots
{\lambda_{z_k}}\uplabel{k}
R_{\ell}\uplabel{k}(z_1,\ldots,z_k).$$
The choice of component $D_{\ell}$ was arbitrary, so a similar
statement holds for $f\uplabel k$ over its whole range, as required by
part~(i)
of the lemma.

Finally,
$$
 \sum_{z\in\eqclass{\alpha}{k}}\lambda\uplabel k_{z}
=c_{\ell}\sum_{z\in\eqclass{\alpha}{k}}
  {\lambda'}_{z}\uplabel k /\bar\lambda_{\alpha}
 =c_{\ell},
$$
establishing part~(iii).
\end{proof}

\begin{lem}
\label{lem:two}
 Let $g:D^r \rightarrow \nnalgs$ be
a symmetric function
with arity $r\geq 3$.
Let $k$ be an integer in $\{3,\ldots,r\}$. Suppose
that $g$ is $(k-1)$-factoring and $(k-1)$-equational.
Suppose there are positive constants
$\{\lambda\uplabel k_{z}:
z\in D\}$   such that
$f\uplabel k(z_1,\ldots,z_k)
=\lambda\uplabel k_{z_1}\cdots\lambda\uplabel k_{z_k}\,
R\uplabel k (z_1,\ldots,z_k)$.
Either $\eval(f\uplabel k)$ is \nump-hard
(which implies that $\eval(g)$ is \nump-hard), or, for every
$\ell\in[m]$,
the multiset
$\{\lambda\uplabel k_{z}:
z\in \eqclass{\alpha}{k}\}$
is independent of the choice of $\alpha\in A_\ell\uplabel k$.
\end{lem}

\begin{proof}
In preparation for the proof, consider
the unary constraint $U(x)$ applied to a variable~$x$
and defined as follows: Take $k-1$ new variables $x_2,\ldots,x_k$
then add the constraint $f\uplabel k(x,x_2,\ldots,x_k)$.
The resulting unary relation $U(x)$  will be used in
the reduction that follows.
For any
$\ell\in[m]$ and $\alpha\in A_{\ell}\uplabel k$,
let $n_{\ell}=|A_{\ell}\uplabel k|$
and
$c_{\ell}=\sum_{z\in\eqclass{\alpha}{k}}\lambda_{z}\uplabel k$
(which, by Lemma~\ref{lem:one},
is independent of the choice of $\alpha\in A_{\ell}\uplabel k$).
For any $z_1\in D_{\ell}$,
\begin{align*}
U(z_1)& =\sum_{z_2,\ldots,z_k\in D_{\ell}}f_{\ell}\uplabel k(z_1,\ldots,z_k)
   =\sum_{z_2,\ldots,z_k\in D_{\ell}}
\lambda\uplabel k_{z_1}\cdots\lambda\uplabel k_{z_k}\,
   R_\ell\uplabel k (z_1,\ldots,z_k) \\
   &=
\sum_{\alpha_2,\ldots,\alpha_k\in A_\ell\uplabel k:
   (\eqrep{z_1}{k},\alpha_2,\ldots,\alpha_k)\in R_\ell \uplabel k}
 \,\,
\sum_{z_2\in\eqclass{\alpha_2}{k},\ldots,z_k\in\eqclass{\alpha_k}{k}}
\lambda\uplabel k_{z_1}\cdots\lambda\uplabel k_{z_k}
      \\
   &=
\lambda\uplabel k_{z_1}
\sum_{\alpha_2,\ldots,\alpha_k\in A_\ell\uplabel k:
   (\eqrep{z_{1}}{k},\alpha_2,\ldots,\alpha_k)\in R_\ell \uplabel k}
 \,\,
\bigg(
\sum_{z_2\in\eqclass{\alpha_2}{k}}
\lambda\uplabel k_{z_2}
\bigg)
\cdots
\bigg(
\sum_{z_k\in\eqclass{\alpha_k}{k}}
\lambda\uplabel k_{z_k}
\bigg)
 \\
   &=\lambda\uplabel k_{z_1} n_{\ell}^{k-2}c_{\ell}^{k-1},
\end{align*}
where the final equality uses part~(ii) of Lemma~\ref{lem:one}.

The idea of the proof is to  use $U$ to ``power up'' vertex weights
$\lambda_{z}\uplabel k$.  In this way we discover that not only is
$\sum_{z\in\eqclass{\alpha}{k}}\lambda_{z}\uplabel k$
independent of $\alpha\in A_{\ell}\uplabel k$,
but so also is
$\sum_{z\in\eqclass{\alpha}{k}}(\lambda_{z}\uplabel k)^j$
for any positive integer~$j$.
This implies that the multiset of weights on an equivalence
class~$\eqclass{\alpha}{k}$ is independent of $\alpha\in A_{\ell}\uplabel k$.

For $z_1,\ldots,z_k\in D_{\ell}$ and $j\geq1$, define
$$
\psi_{z_1} =( \lambda\uplabel k_{z_1} n_{\ell}^{k-2}c_{\ell}^{k-1}
 )^{j-1} \lambda\uplabel k_{z_1}
$$
and
$$
h_\ell\uplabel j(z_1,\ldots,z_k) =
 \psi_{z_1}\cdots\psi_{z_k} R_\ell\uplabel k (z_1,\ldots,z_k).
$$
Let $h\uplabel j = h_1\uplabel j \oplus \cdots \oplus h_m\uplabel j$.
We will give a reduction from $\eval(h\uplabel j)$ to 
$\eval(f\uplabel k)$.
Suppose $G=(V,E)$
is a $k$-uniform hypergraph (an input to $\eval(h\uplabel j)$).
For $j\geq 1$,
the hypergraph $G \uplabel j$
is obtained from $G$ as follows:  for each vertex~$v$ in~$G$
of degree $d_{v}$, add $(k-1)(j-1)d_{v}$ new vertices and $(j-1)d_{v}$ new
edges, each one incident at $v$ and at $k-1$ of the new vertices.
Then
\begin{align*}
Z^{h\uplabel j_\ell}(G)&=\sum_{\sigma:V\to D_{\ell}}\,
   \prod_{(u_1,\ldots,u_k)\in E}h\uplabel
   j_{\ell}(\sigma(u_1),\ldots,\sigma(u_k))\\
&=\sum_{\sigma:V\to D_{\ell}}\,
   \prod_{(u_1,\ldots,u_k)\in E}
 \psi_{\sigma(u_1)}\cdots\psi_{\sigma(u_k)}\,
    R_{\ell}\uplabel k(\sigma(u_1),\ldots,\sigma(u_k))\\
&=\sum_{\sigma:V\to D_{\ell}}\,
   \prod_{v\in V} (
\lambda\uplabel k_{\sigma(v)} n_{\ell}^{k-2}c_{\ell}^{k-1}
  )^{(j-1)d_{v}}\!\!
   \prod_{(u_1,\ldots,u_k)\in E}
\lambda\uplabel k_{\sigma(u_1)}\cdots \lambda \uplabel k_{\sigma(u_k)}\,
    R_{\ell}\uplabel k(\sigma(u_1),\ldots,\sigma(u_k))\\
&=\sum_{\sigma:V\to D_{\ell}}\,
   \prod_{v\in V}(
\lambda\uplabel k_{\sigma(v)} n_{\ell}^{k-2}c_{\ell}^{k-1}
  )^{(j-1)d_{v}}\!\!
  \prod_{(u_1,\ldots,u_k)\in E}
    f_{\ell}\uplabel k(\sigma(u_1),\ldots,\sigma(u_k))\\
 &=Z^{f\uplabel k_\ell}({G \uplabel j}).
\end{align*}

Thus (for connected $G$)
$$
Z^{h\uplabel j}(G)=
  \sum_{\ell\in[m]}Z^{h\uplabel j_\ell}(G)=
  \sum_{\ell\in[m]}Z^{f\uplabel k_\ell}(G \uplabel j)=
  Z^{f \uplabel k}(G \uplabel j),
$$
so $\eval(h\uplabel j)\leq\eval(f\uplabel k)$.

Assume $\eval(f\uplabel k)$ is not \nump-hard.  Then $\eval(h \uplabel j)$ is
not
\nump-hard
for any $j\geq1$.
Recall from the statement of the lemma that $g$ is $(k-1)$-factoring and $(k-1)$-equational.
Then
from Lemma~\ref{lem:one} part (iii),
$$
\sum_{z\in\eqclass{\alpha}{k}}\psi_{z}
   =(n_{\ell}^{k-2}c_{\ell}^{k-1})^{j-1}
   \sum_{z\in\eqclass{\alpha}{k}}(\lambda_{z}\uplabel k)^j
$$
is independent of $\alpha\in A_{\ell}\uplabel k$ for all $j\geq1$.
This can only occur if the multiset $\{\lambda_{z}\uplabel
k:z\in\eqclass{\alpha}{k}\}$ is
independent of $\alpha\in A_{\ell}\uplabel k$.
\end{proof}

We will use the following corollary of Lemmas~\ref{lem:base},
\ref{lem:one} and~\ref{lem:two}.
\begin{cor}\label{cor:cor}
Let $g:D^r \rightarrow \nnalgs$ be
a symmetric function
with arity $r\geq 3$.
Let $k$ be an integer in $\{3,\ldots,r\}$. Suppose
that $g$ is $(k-1)$-factoring and $(k-1)$-equational.
Either $\eval(f\uplabel k)$ is \nump-hard
(which implies that $\eval(g)$ is \nump-hard), or
$g$ is $k$-factoring.
\end{cor}
\begin{proof}
By Lemma~\ref{lem:one}
part (i)
there are positive constants
$\{\lambda\uplabel k_{z}:
z\in D\}$
such that
$$f\uplabel k(z_1,\ldots,z_k)
=\lambda\uplabel k_{z_1}\cdots\lambda\uplabel k_{z_k} R\uplabel k
(z_1,\ldots,z_k).$$
Fix any $\ell\in[m]$. By Lemma~\ref{lem:two},
the multiset
$\{\lambda\uplabel k_{z}:
z\in \eqclass{\alpha}{k}\}$
is independent of the choice of $\alpha\in A_\ell\uplabel k$
Let $s_\ell\uplabel k$ be the size of this multiset.
Then $D_{\ell}\cong A_{\ell}\uplabel k\times [s_{\ell}\uplabel k]$ giving
condition (1) in
the definition of $k$-factoring.
Also, if the element $z\in D_\ell$ corresponds to the $i$'th element of
the $\sim_k$ class $\eqclass{z}{k}$  then the value $\lambda\uplabel k_z$ just
depends upon $i$ (and on $\ell$) --- it is independent of
the equivalence class~$\eqclass{z}{k}$.
We denote this value as $\lambda_{\ell,i}\uplabel k$.
Thus, for
 $\alpha_1,\ldots,\alpha_k \in A_\ell\uplabel k$ and
$i_1,\ldots,i_k\in [s_\ell\uplabel k]$,
$$
f\uplabel{k}_{\ell}( (\alpha_1,i_1),\ldots,(\alpha_k,i_k) ) =
   \lambda\uplabel{k}_{\ell,i_1}\cdots\lambda\uplabel{k}_{\ell,i_k}
R\uplabel{k}_\ell(\alpha_1,\ldots,\alpha_k),
$$
giving condition (2) in the definition of $k$-factoring.
\end{proof}

\begin{lem} \label{cor:one}
Let $g:D^r \rightarrow \nnalgs$ be
a symmetric function
with arity $r\geq 3$.
Let $k$ be an integer in $\{3,\ldots,r\}$. Suppose
that $g$ is
$k$-factoring.
Then, for every $\ell\in[m]$,
$$
Z^{f_\ell\uplabel k}(G)=\Lambda_{\ell}\uplabel k(G) \,\,Z^{S_\ell\uplabel
k}(G),
$$
 where
\begin{equation}\label{eq:90}
\Lambda_{\ell}\uplabel k(G)=
\prod_{v\in V(G)}\,
   \sum_{i\in[s_{\ell}\uplabel k]}(\lambda_{\ell,i}\uplabel k)^{d_{v}}.
\end{equation}
\end{lem}

\begin{proof}

For $G=(V,E)$,
\begin{align*}
Z^{f_\ell\uplabel k}(G)&=
   \sum_{\sigma:V\to A_{\ell}\uplabel k,\tau:V\to [s_\ell\uplabel k]}\,
   \prod_{(u_1,\ldots,u_k)\in E}f_{\ell}\uplabel
   k((\sigma(u_1),\tau(u_1)),\ldots,(\sigma(u_k),\tau(u_k)))\\
&=\sum_{\sigma:V\to A_{\ell}\uplabel k,\tau:V\to[s_\ell\uplabel k]}\,
   \prod_{(u_1,\ldots,u_k)\in E}
   \lambda_{\ell,\tau(u_1)}\uplabel k
   \cdots
   \lambda_{\ell,\tau(u_k)}\uplabel k\,
   S_{\ell}\uplabel k(\sigma(u_1),\ldots,\sigma(u_k))\\
&=\sum_{\sigma:V\to A_{\ell}\uplabel k}
\bigg(
 \prod_{(u_1,\ldots,u_k)\in E}
S_{\ell}\uplabel k(\sigma(u_1),\ldots,\sigma(u_k))
\bigg)
\bigg(
 \sum_{\tau:V\to[s_\ell\uplabel k]}\,
   \prod_{v\in V}
   \big(\lambda_{\ell,\tau(v)}\uplabel k\big)^{d_v}
   \bigg) \\
&= Z^{S_\ell\uplabel k}(G)\,\,\Lambda_{\ell}\uplabel k(G).
\end{align*}

\end{proof}

\begin{lem} \label{cor:two}
Let $g:D^r \rightarrow \nnalgs$ be
a symmetric function
with arity $r\geq 3$.
Let $k$ be an integer in $\{3,\ldots,r\}$. Suppose
that $g$ is $(k-1)$-factoring and $(k-1)$-equational.
Either $\eval(f\uplabel k)$ is \nump-hard
(which implies that $\eval(g)$ is \nump-hard), or
$\eval(S\uplabel k) \leq \eval(f\uplabel k)$.
\end{lem}
\begin{proof}
Suppose that $G$ is a connected $k$-uniform hypergraph.
For any positive integer, $p$,
let $G^1,\ldots,G^p$ be copies of $G$.
Let $\{v_{1}^j,\ldots,v_{n}^j\}$ be the vertices of $G^j$.
Construct $G\uplabel p$ by taking the union of $G^1,\ldots,G^p$
along with $n(k-1)p$ new vertices
and $2 n p$ new edges:
For each $i\in[n]$, $t\in[k-1]$ and $j\in[p]$
we add a vertex $u_{i,t}^j$.
Then we add edges
$(u_{i,1}^j,\ldots,u_{i,k-1}^j,v_i^j)$
and
$(u_{i,1}^j,\ldots,u_{i,k-1}^j,v_i^{(j\bmod n)+1})$.

Now
by Corollary~\ref{cor:cor}, $g$ is $k$-factoring, so
$D_{\ell}\cong A_{\ell}\uplabel k\times [s_{\ell}\uplabel k]$.
By Lemma~\ref{cor:one},
 \begin{equation}
\label{eq:ZGp}
Z^{f\uplabel k}(G\uplabel p)=
\sum_{\ell\in[m]}
\Lambda_{\ell}\uplabel k(G\uplabel p) \,\,Z^{S_\ell\uplabel k}(G\uplabel p).
\end{equation}

We now look at the constituent parts of the right-hand-side of
Equation~(\ref{eq:ZGp}).
First,
$$Z^{S_\ell\uplabel k}(G\uplabel p) =
\sum_{\sigma: V(G\uplabel p)\rightarrow A_\ell\uplabel k}\,
\prod_{(\U_1,\ldots,\U_k)\in E(G\uplabel p)}
S_\ell\uplabel k(\sigma(\U_1),\ldots,\sigma(\U_k)).$$
By Part (ii) of Lemma~\ref{lem:one}, $S\uplabel k_\ell$
is a Latin hypercube. So, given the values
$\sigma(v_1^j),\ldots,\sigma(v_n^j)$,
the values $\sigma(u_{i,1}^j),\ldots,\sigma(u_{i,k-2}^j)$
(for $i\in[n]$)
can be chosen arbitrarily from $A_\ell\uplabel k$.
Then
there is exactly one choice for each $\sigma(u_{i,k-1}^j)$ so
that
$$(\sigma(u_{i,1}^j),\ldots,\sigma(u_{i,k-1}^j),\sigma(v_i^j))\in
S_\ell\uplabel k.$$
Then for $j<n$ to have
$$(\sigma(u_{i,1}^j),\ldots,\sigma(u_{i,k-1}^j),\sigma(v_i^{(j\bmod n)+1}
))\in S_\ell\uplabel k$$
we must have $\sigma(v_i^{j+1})=\sigma(v_i^j)$.
(If $j={n}$ then
$$ (\sigma(u_{i,1}^j),\ldots,\sigma(u_{i,k-1}^j),\sigma(v_i^{(j\bmod n)+1}))
\in S_\ell\uplabel k$$ just ensures $v_i^1=v_i^n$
so it
adds no new constraint.)
Thus,
\begin{align*}
Z^{S_\ell\uplabel k}(G\uplabel p) &=
\sum_{\sigma: V(G^1)\rightarrow A_\ell\uplabel k}
\prod_{(\U_1,\ldots,\U_k)\in E(G^1)}
S_\ell\uplabel k(\sigma(\U_1),\ldots,\sigma(\U_k))
\,\bigl|A_\ell\uplabel k\bigr|^{n(k-2)p}\\
&=
\bigl|A_\ell\uplabel k\bigr|^{n(k-2)p} Z^{S_\ell\uplabel k}(G).
\end{align*}

Also, using $d_\Gamma(\U)$ to denote the degree of vertex~$\U$ in
hypergraph~$\Gamma$,
 \begin{align*}
\Lambda_{\ell}\uplabel k(G\uplabel p)
&=
\prod_{\U\in V(G\uplabel p)}\,
   \sum_{h\in[s_{\ell}\uplabel k]}
   \big(\lambda_{\ell,h}\uplabel k\big)^{d_{G\uplabel p}(\U)}
\\
&=\left(
 \prod_{i\in[n]}\,
   \sum_{h\in[s_{\ell}\uplabel k]}\big(\lambda_{\ell,h}\uplabel
   k\big)^{d_{G}(v_i)+2}
\right)^p
\left(
 \prod_{i\in[n]}\,\prod_{t\in[k-1]}\,
\sum_{h\in[s_{\ell}\uplabel k]}\big(\lambda_{\ell,h}\uplabel k\big)^2
\right)^p,\\
\end{align*}
where the first factor on the right-hand-side is the product
over vertices $v_i^j$ and the
second factor is the product over vertices $u_{i,t}^j$.

So $Z^{f\uplabel k}(G\uplabel p)$ is equal to
$$
\sum_{\ell\in[m]}
\bigg(
 \prod_{i\in[n]}\,
   \sum_{h\in[s_{\ell}\uplabel k]}
   \big(\lambda_{\ell,h}\uplabel k\big)^{d_{G}(v_i)+2}
\bigg)^p
\bigg(
 \prod_{i\in[n]}\,\prod_{t\in[k-1]}\,
\sum_{h\in[s_{\ell}\uplabel k]}\big(\lambda_{\ell,h}\uplabel k\big)^2
\bigg)^p
\bigl|A_\ell\uplabel k\bigr|^{n(k-2)p} Z^{S_\ell\uplabel k}(G).
$$
We can now use Corollary~\ref{cor:interp} with
$Z_p = Z^{f\uplabel k}(G\uplabel p)$,
$\gamma_\ell = Z^{S_\ell\uplabel k}(G)$
and
\[\eta_\ell =
\bigg(
 \prod_{i\in[n]}\,
   \sum_{h\in[s_{\ell}\uplabel k]}\big(\lambda_{\ell,h}\uplabel
   k\big)^{d_{G}(v_i)+2}
\bigg)
\bigg(
 \prod_{i\in[n]}\,\prod_{t\in[k-1]}\,
\sum_{h\in[s_{\ell}\uplabel k]}\big(\lambda_{\ell,h}\uplabel k\big)^2
\bigg)
\bigl|A_\ell\uplabel k\bigr|^{n(k-2)}
.\qedhere\]
\end{proof}

Let us take stock.  Suppose $g$ is not \nump-hard
and that $g$ is $(k-1)$-factoring and $(k-1)$-equational.
We know by Corollary~\ref{cor:cor} that $g$ is $k$-factoring, and
by Part~(ii) of Lemma~\ref{lem:one}
that the
various relations $S_{\ell}\uplabel k$ are Latin hypercubes.
The final step, the subject of the following section,
is to show that the latter have additional structure, namely
that they are defined by equations over an Abelian groups.
It will follow that $g$ is $k$-equational.

\section{Constraint satisfaction and Abelian group equations}
\label{sec:triple}

Let $S$ be an arity-$k$ relation on a ground set $A$.
Recall our earlier discussion, in Section~\ref{sec:intro},
on the
relation between $\eval(S)$ and $\ncsp(S)$.
Every instance $G$ of $\eval(S)$ can be viewed as an instance of $\ncsp(S)$
by taking the vertices as variables and the edges as constraint scopes.
However, we noted that the converse is not true, since
an instance $I$ of $\ncsp(S)$ might not be
a properly-formed instance of $\eval(S)$.
Nevertheless, by copying variables,
we can view an instance~$I$ of $\ncsp(S)$ as being a $k$-uniform
hypergraph~$G$,
together with some binary equality constraints on variables.
For variables $U$ and $W$,
the constraint $=(U,W)$ is satisfied if and only if
$\sigma(U)=\sigma(W)$.
The following lemma shows that, in our setting,  these equality
constraints do not add any real power - they can be implemented by
interpolation.

\begin{lem}
Let $S=S_1\oplus \cdots \oplus S_m$ be a symmetric
$k$-ary relation on a ground set $A$,
such that each $S_\ell$ is a Latin hypercube.
Then $\ncsp(S) \leq \eval(S)$.
\label{lem:anotherboringinterpolation}
\end{lem}

\begin{proof}
For $\ell\in[m]$,
let $A_\ell$ be the ground set of $S_\ell$.

Let $I$ be an instance of $\ncsp(S)$ comprising a connected hypergraph $G$
with vertices $\{v_1,\ldots,v_n\}$
and $\nu$ equality constraints. Note that this is without loss of generality --
an instance $I$
may be represented as a hypergraph $G$ together with equality constraints
in which equality is only applied to variables in the same connected component
of $G$.

For a positive integer~$p$, construct a  hypergraph $G\uplabel p$ by combining
$G$ with
$\nu p (k-1)$ new vertices and $2 \nu p$ new edges:
For $j\in[p]$ and $i\in[\nu]$ add vertices 
$u_{i,1}^j,\ldots,u_{i,k-1}^j$.
If the $i$'th equality constraint is $=(v_s,v_t)$
then add the $2 p$ edges
$(v_s,u_{i,1}^j,\ldots,u_{i,k-1}^j)$
and $(v_t,u_{i,1}^j,\ldots,u_{i,k-1}^j)$ for $j\in[p]$.

Now, suppose we are given the values $\sigma(v_1),\ldots,\sigma(v_n)$ in
$A_\ell$.
By the Latin hypercube property, we can have
$(\sigma(v_s),\sigma(u_{i,1}^j),\ldots,\sigma(u_{i,k-1}^j))\in S$
and
$(\sigma(v_t),\sigma(u_{i,1}^j),\ldots,\sigma(u_{i,k-1}^j))\in S$
only if $\sigma(v_s)=\sigma(v_t)$.
In that case, there are
$|A_\ell|^{k-2}$ choices for
$\sigma(u_{i,1}^j),\ldots,\sigma(u_{i,k-1}^j)$.
So
$$Z^S(G\uplabel p) = \sum_{\ell\in [m]} Z^{S_\ell}(I) |A_\ell|^{(k-2)p}.$$
We can now use Corollary~\ref{cor:interp}.
\end{proof}
The following lemma establishes the algebraic structure of the $S_\ell$, using
a result of Bulatov and Dalmau~\cite{BulDal07}. The proof itself has
similarities to
that of P\'alfy's theorem~\cite{Palfy84} (see, for example, \cite{DenWis02}).
\begin{lem}
Suppose $k\geq 3$.
Let $S=S_1\oplus \cdots \oplus S_m$ be a symmetric $k$-ary
relation on a ground set $A$
such that, for each $\ell\in[m]$,
$S_\ell$ is a  Latin hypercube.
Suppose $\eval(S)$ is not \nump-hard.
Then for each $\ell\in[m]$,
the relation $S_\ell$ is  defined by an equation
over an Abelian
group $\G_\ell=\langle A_{\ell},+\rangle$ as follows:
for some element $a \in A_{\ell}$,
$(\alpha_1, \ldots, \alpha_k) \in S_{\ell}$  if and only if
$\alpha_1 + \cdots + \alpha_k = a$.
\label{lem:three}
\end{lem}

\begin{proof}
Suppose $\eval(S)$ is not \nump-hard. Fix $\ell\in [m]$, and fix any
element $a_\ell\in A_\ell$ and denote it by $0$. If
$(\alpha,\beta,\gamma,0,\ldots,0)\in S_\ell$ we will write
$\gamma=\alpha\cdot \beta$. Then we will call $(\alpha,\beta,\gamma)$
a \emph{triple} and denote the set of triples by $T_\ell$. We will
call $(\alpha,\beta,\gamma,0,\ldots,0)\in S_\ell$ the corresponding
\emph{padded triple}. For given $\alpha$ and $\beta$, the existence
and uniqueness of $\gamma$ in a padded triple follows directly from
the fact that $S_\ell$ is a Latin hypercube.
Thus we may regard $\alpha\cdot \beta$ as a binary operation on $A_\ell$,
and hence $\A_\ell=\langle A_\ell,\cdot\rangle$ is an algebra.
By symmetry, the binary operation of $\A_\ell$ is commutative,
and satisfies the identity $\alpha\cdot(\alpha\cdot \beta)=\beta$ for all
$\alpha,\beta\in A_\ell$. However, the operation is not necessarily
associative.

By Lemma~\ref{lem:anotherboringinterpolation},  $\ncsp(S) \leq \eval(S)$, so
$\ncsp(S)$ is not \nump-hard.
Thus, by~\cite{BulDal07}, there is a
\maltsev polymorphism $\varphi(\alpha,\beta,\gamma)$ on $A$ which preserves
$S$.
Recall that a \maltsev operation $\varphi:A^3\to A$ is any function which
satisfies the identities
$\varphi(\alpha,\beta,\beta)=\varphi(\beta,\beta,\alpha)=\alpha$ for all
$\alpha,\beta\in A$. We may use $\varphi$ to calculate, as follows.
Each line of a table is a triple in $T_\ell$,
and the \maltsev polymorphism implies that the bottom line is also a triple in
$T_\ell$, using the fact that $\varphi(0,0,0)=0$ in the padded triples
(which follows from the \maltsev property).
Thus
\[\begin{array}{c@{\hspace{1.25cm}}c@{\hspace{1.25cm}}c}
    \alpha & \gamma & \alpha\cdot \gamma \\
    \beta & \gamma & \beta\cdot \gamma \\
    \gamma & \beta & \beta\cdot \gamma \\ \hline
    \varphi(\alpha,\beta,\gamma) & \beta & \alpha\cdot \gamma
\end{array}\]
and hence $\varphi(\alpha,\beta,\gamma)=\beta\cdot(\alpha\cdot \gamma)$ is a
term of the algebra $\A_\ell$. We have
\[\varphi(\alpha,\beta,\gamma)=\beta\cdot(\alpha\cdot
\gamma)=\beta\cdot(\gamma\cdot \alpha)=\varphi(\gamma,\beta,\alpha),\]
so $\varphi$ is a symmetric \maltsev operation
(in the sense that it is symmetric in the first and third arguments).

Define a new binary operation $+$ on $A_\ell$ by
$\alpha+\beta=\varphi(\alpha,0,\beta)=0\cdot(\alpha\cdot \beta)$. It follows
immediately that $+$ is commutative. Hence
\[ 0+\alpha\ =\ \alpha+0\ =\ 0\cdot(\alpha\cdot 0)\ 
= \varphi(\alpha,0,0) 
=\ \alpha,\]
so $0$ is an identity for $+$. Denote $0\cdot 0$ by $0^2$, and define $-\alpha$
by $\alpha\cdot 0^2$. Then
\[ (-\alpha)+\alpha\ =\ \alpha+(-\alpha)\ =\ 0\cdot(\alpha\cdot(\alpha\cdot
0^2))\ =\ 0\cdot(0^2)\
=\ 0\cdot(0\cdot 0)\ =\ 0,\]
so $-\alpha$ is an inverse for $\alpha$. As usual, we write $\alpha-\beta$ for
$\alpha+(-\beta)$.

We have
\[\begin{array}{c@{\hspace{1.25cm}}c@{\hspace{1.25cm}}c}
    \alpha & 0^2\ & \alpha\cdot 0^2  \\
    0 & 0^2\ & 0  \\
    \beta & \beta\cdot 0 & 0  \\ \hline
    \alpha+\beta& \beta\cdot 0 & \alpha\cdot 0^2 \end{array}\]
so $\alpha+\beta=(\beta\cdot 0)\cdot(\alpha\cdot 0^2)$
and since $+$ is commutative,
$\alpha+\beta=\beta+\alpha = (\alpha\cdot 0)\cdot(\beta \cdot 0^2)$.
Then
\[\begin{array}{c@{\hspace{1.25cm}}c@{\hspace{1.25cm}}c}
\alpha \cdot 0 & \beta \cdot 0^2 & \alpha+\beta \\
0^2 & 0 & 0 \\
\gamma\cdot 0 & 0 & \gamma  \\ \hline
\varphi(\alpha \cdot 0,0^2,\gamma\cdot 0)& \ \beta\cdot 0^2  &
(\alpha+\beta)+\gamma
\end{array}\]
Therefore,
since $\varphi$ is symmetric in its first and third arguments,
\begin{align*}
    (\alpha+\beta)+\gamma\ &=\ \varphi(\alpha \cdot 0,0^2,\gamma\cdot
    0)\cdot(\beta\cdot 0^2) \ =\ \varphi(\gamma \cdot 0,0^2,\alpha\cdot
    0)\cdot(\beta\cdot 0^2)\\ &=\ (\gamma+\beta)+\alpha\ =\
    \alpha+(\gamma+\beta)
    \ =\ \alpha+(\beta+\gamma).
\end{align*}
The operation $+$ is therefore associative, and hence the algebra
$\G_\ell=\langle A_\ell,+,-,0\rangle$ is an Abelian group. Hence,
since 
$- X$ is defined to be $X \cdot 0^2$ and
$\alpha-0^2 = -(-\alpha+0^2)$,
we have, for any $\alpha,\beta\in A_\ell$,
\[\begin{array}{c@{\hspace{1.25cm}}c@{\hspace{1.25cm}}c}
\alpha-0^2 & 0^2 & -\alpha+0^2 \\
0\phantom{.} & 0^2  & \phantom{-}0 \\
0^2 & \beta & -\beta  \\ \hline
\alpha & \beta  & -\alpha-\beta+0^2
\end{array}\]
where we used the fact that, by definition,
$\varphi( x, 0, y) = x + y$.
Thus $\alpha\cdot \beta=-\alpha-\beta+0^2.$, and it follows that
\begin{align}
T_\ell\ &=\ \set{(\alpha,\beta,-\alpha-\beta+0^2)\in
A_\ell^3:\alpha,\beta\in\A_\ell}\notag\\
&=\ \set{(\alpha,\beta,\gamma)\in A_\ell^3:\alpha+\beta+\gamma=0^2\textrm{ in
}\G_\ell}.\label{tripleset}
\end{align}
In particular, $(\alpha,-\alpha,0^2)\in T_\ell$ for all $\alpha\in A_\ell$, and
hence $(0,0,0^2)\in T_\ell$. It  follows further that
\[ \varphi(\alpha,\beta,\gamma)=\beta\cdot(\alpha\cdot \gamma)=
-\beta-(\alpha\cdot
\gamma)+0^2=-\beta-(-\alpha-\gamma+0^2)+0^2=\alpha-\beta+\gamma, \]
so the \maltsev operation is the term $\alpha-\beta+\gamma$ in the Abelian
group $\G_\ell$.

Now assume by induction that the conclusion of the lemma is true for any
$S$ of arity less than~$k$. It is true for arity $3$  by \eqref{tripleset},
since then,
for any $\ell\in[m]$,
$S_\ell=T_\ell$. For larger $k$, suppose
$(\alpha_1,\alpha_2,\ldots,\alpha_k)\in S_\ell$ is arbitrary. Then, using the
\maltsev operation and padding the triples $(\alpha_1,-\alpha_1,0^2)$,
$(0,0,0^2)$, we have
\[\begin{array}{*6{c@{\hspace{0.75cm}}}}
\alpha_1 & \phantom{-}\alpha_2 & \alpha_3 & \alpha_4 & \cdots & \alpha_k \\
\alpha_1 & -\alpha_1 & 0^2 & 0 & \cdots & 0 \\
0 & \phantom{-}0 & 0^2 & 0 & \cdots & 0 \\ \hline
0 & \alpha_1+\alpha_2 & \alpha_3 & \alpha_4 & \cdots & \alpha_k
\end{array}\]
Now the $(k-1)$-ary relation
\[ S'_\ell\ =\ \{(\alpha'_2,\alpha'_3,\ldots,\alpha'_k)\in
A_\ell^{k-1}:(0,\alpha'_2,\alpha'_3,\ldots,\alpha'_k)\in S_\ell\}\]
is symmetric and has the same \maltsev operation as $S_\ell$. Thus we
can define the same Abelian group $\G_\ell$, and by induction we will have
\[\textstyle S'_\ell\ =\ \{(\alpha'_2,\alpha'_3,\ldots,\alpha'_k)\in
A_\ell^{k-1}:\ \sum_{j=2}^{k}\alpha'_j=a'\textrm{\ \,in }\G_\ell\},\]
for some $a'\in A_\ell$. But we have shown that, for all
$(\alpha_1,\alpha_2,\alpha_3,\ldots,\alpha_k)\in S_\ell$, we have
$(\alpha_1+\alpha_2,\alpha_3\ldots,\alpha_k)\in S'_\ell$. Thus, since $\G_\ell$
is an Abelian group,
\[\textstyle S_\ell\ =\ \{(\alpha_1,\alpha_2,\alpha_3,\ldots,\alpha_k)\in
A_\ell^{k}:\ \sum_{j=1}^{k}\alpha_j=a\textrm{\ \,in }\G_\ell\},\]
where $a=a'$, completing the induction and the proof.
\end{proof}

\section{Proof of Theorem~\ref{thm:main}}
\label{sec:final}

\begin{proof}
Let $g:D^r \rightarrow \nnalgs$ be
a symmetric function
with arity $r\geq 3$.
First, suppose that $g$ is $r$-factoring and $r$-equational.
Then applying Lemma~\ref{cor:one} with $k=r$,
we find that, for connected $G$,
\begin{equation}\label{eq:80}
Z^{g}(G)=\sum_{\ell\in[m]}\Lambda_{\ell}\uplabel r (G) \,\,Z^{S_\ell\uplabel
r}(G).
\end{equation}
Now since $g$ is $r$-equational,
$S\uplabel{r}_\ell$ is defined by an equation over an Abelian group
$(A\uplabel{r}_\ell,+)$.
Now, by~\cite[Lemma 13]{KlLaTe06}, $\eval(S_\ell\uplabel r)$ is polynomial time
solvable: The Abelian group is a direct product of cyclic groups of prime
power. For each of these cyclic groups, we just need to count the solutions
to a system of linear equations over the field $\mathbb{Z}_p$ and
this can be done in polynomial time (see \cite{KlLaTe06}).
Thus, $\eval(S_\ell\uplabel r)$ is in \fp. To
show that $\eval(g)$ is in \fp, it remains to show that
$\Lambda_\ell\uplabel r(G)$, as defined in~(\ref{eq:90}),
can be computed in \fp. This is immediate over the number
field $\rats(\theta,\lambda_{\ell,1}\uplabel r,\ldots,\lambda_{\ell,s_\ell}\uplabel r)$. In Section~\ref{sec:modelcomp},
we show that it can even be computed in \fp
over the number field $\rats(\theta)$.

Suppose now that
$\eval(g)$ is not \nump-hard.  Then
by Lemma~\ref{lem:base}, $g$ is
both $2$-factoring and $2$-equational.
Next suppose that, for some $k\in\{3,\ldots,r\}$,
$g$ is $(k-1)$-factoring and $(k-1)$-equational.
Since $\eval(g)$ is not \nump-hard, we know that
$\eval(f\uplabel k)$ is not \nump-hard.
By Corollary~\ref{cor:cor},
$g$ is $k$-factoring.
Suppose, for contradiction, that
$g$ is not $k$-equational.
By Part (ii) of Lemma~\ref{lem:one}, each
$S\uplabel k_\ell$
is a Latin hypercube, so by Lemma~\ref{lem:three}, $\eval(S\uplabel k)$
is \nump-hard. By Lemma~\ref{cor:two},
$\eval(f\uplabel k)$ is \nump-hard, giving the contradiction.
So $g$ is $k$-equational.
By induction, $g$ is $r$-factoring and $r$-equational.

It remains to consider the effectiveness of the dichotomy.
For this, we must show that there is an algorithm that
determines whether $g$ is $r$-factoring and $r$-equational.
This is nearly identical to a proof that the dichotomy
in Theorem~\ref{thm:state} is effective, however the
notation is simpler in the latter context, so we provide this proof next.
\end{proof}

\begin{lem} The dichotomy in Theorem~\ref{thm:state} is effective.
\end{lem}
\begin{proof}

We must show that there is an algorithm that determines
whether the conditions in Theorem~\ref{thm:state} are satisfied.
The connected components $D_1,\ldots, D_m$ can easily be
determined. Then, for each $\ell\in[m]$, there are a constant number of
possibilities for the decompositions
$D_\ell \cong A_\ell \times [s_\ell]$ ($\ell\in[m]$) which can all be checked, if
necessary. Then, for the third condition, there are only a finite number of possibilities for the group structure, corresponding to the factorisations of $|A_\ell|$. Again, these can all be checked to see if any defines $S_\ell$,
for each $\ell\in[m]$.

For the second condition,  for each $\ell\in[m]$,
we
need to decide the satisfiability of a system of the form
\begin{equation}\label{eq:40}
g( (\alpha_1,i_1),\ldots,(\alpha_r,i_r) ) =
   \lambda_{\ell,i_1}\cdots\lambda_{\ell,i_r}\ \textrm{ for all }\
(\alpha_1,\ldots,\alpha_r)\in S_\ell\ \textrm{ and }\ i_1,\ldots,i_r\in[s_\ell].
\end{equation}
Thus we have
\begin{equation}\label{eq:50}
   \lambda_{\ell,i}\ =\ g( (\alpha_1,i),\ldots,(\alpha_r,i) )^{1/r}\
   \ \textrm{ for all }\
(\alpha_1,\ldots,\alpha_r)\in S_\ell\ \textrm{ and }\ i\in[s_\ell],
\end{equation}
and hence~\eqref{eq:40} is equivalent to the system
\begin{equation*}
g( (\alpha_1,i_1),\ldots,(\alpha_r,i_r) )^r\  =\ \prod_{j=1}^r
g( (\alpha_1,i_j),\ldots,(\alpha_r,i_j) )
\end{equation*}
for all
$(\alpha_1,\ldots,\alpha_r)\in S_\ell$
and
$ i_1,\ldots,i_r\in[s_\ell]$,
which can be decided in constant time by computation in the number field $\rats(\theta)$.
\end{proof}
\section{Computation of $Z^g(G)$ in $\rats(\theta)$}
\label{sec:modelcomp}

Observe that \eqref{eq:90}, \eqref{eq:80} and \eqref{eq:50} seem together to imply that, in the polynomial time computable cases, we must compute $Z^g(G)$ in the number field $\rats(\theta,\lambda_{1,1},\ldots,\lambda_{1,s_1},\ldots,\lambda_{m,1},\ldots,
\lambda_{m,s_m})$, where, for $\ell\in[m]$ and $i\in[s_\ell]$,
$\lambda_{\ell,i}=\lambda_{\ell,i}\uplabel r$ is an  $r^\textrm{th}$ root
of one of the original weights.
This seems anomalous, since $Z^g(G)$ is actually an element of $\rats(\theta)$. We conclude by showing that the computation of $Z^g(G)$ can be done entirely within $\rats(\theta)$, as might be hoped.

To do this, we must expand the expressions

\begin{equation*}
\Lambda_{\ell}\uplabel r(G)=
\prod_{v\in V(G)}\,
   \sum_{i=1}^{s_{\ell}}(\lambda_{\ell,i})^{d_{v}}.
\end{equation*}

To simplify the text, we drop the subscript $\ell$ in the
rest of this section, writing $s$ for $s_\ell$ and $\lambda_i$ for $\lambda_{\ell,i}$ and
$\Lambda\uplabel r$ for $\Lambda_\ell\uplabel r$. Thus, we wish to expand
\begin{equation*}
    \Lambda\uplabel r (G)\ =\ \prod_{v\in V(G)}\Big(\sum_{i=1}^s \lambda_i^{d_v}\Big).
\end{equation*}

The exponents of $\lambda_i$ ($i\in[s]$) in the monomials of the expansion of $\Lambda\uplabel r (G)$ are given by
\begin{equation}\label{eq:60}
    \sum_{v\in V(G)}\delta_{v,i}d_v,\ \ \textrm{where}\ \sum_{i=1}^s \delta_{v,i}=1\ \ \textrm{and}\ \ \delta_{v,i}\in\set{0,1}\ \ (i\in[s],v\in V(G)).
\end{equation}
Recall that~$M$ denotes the number of edges of~$G$. 
Thus there are $O(M^s)$ possible monomials in the $\lambda_i$, and the integer coefficient of each monomial $\prod_{i=1}^s\lambda_i^{M_i}$ are given by computing the number of solutions to systems of equations of the form
\begin{equation}\label{eq:70}
    \sum_{v\in V(G)}\delta_{v,i}d_v\,=\,M_i,\ \ \textrm{where}\ \sum_{i=1}^s \delta_{v,i}=1\ \ \textrm{and}\ \ \delta_{v,i}\in\set{0,1}\ \ (i\in[s],v\in V(G)).
\end{equation}
This can be done for all $0\leq M_i \leq rM$ ($i\in[s]$) in $O(nM^s)$ time by dynamic programming. An easy counting argument shows that $\sum_{v\in V(G)} d_v=rM$, so this returns a nonzero coefficient for the monomial $\prod_{i=1}^s\lambda_i^{M_i}$ only if $\sum_{i=1}^s M_i=rM$. Thus, in fact, there are at most
\begin{equation*}
    \binom{rM+s-1}{s-1}\ =\ O(M^{s-1})
\end{equation*}
such monomials, which is clearly polynomial in the input size.

Thus we can compute in \fp a representation of $\Lambda\uplabel r (G)$
as a multivariate polynomial with monomials $\prod_{i=1}^s\lambda_i^{M_i}$ such that $\sum_{i=1}^s M_i=rM$ and $M_i\geq 0$ ($i\in[s]$). We can express each such monomial in terms of the original weights, as follows. Let $r_{ij}$ ($i\in[s],j\in [M]$) be nonnegative integers such that $\sum_{i=1}^s r_{ij}=r$ ($j\in[M]$) and $\sum_{j=1}^M r_{ij}=M_i$ ($i\in[s]$). Such numbers always exist, though they will usually be far from unique, and can be computed in $O(M)$ time. They are the entries of a \emph{contingency table} with row totals $M_i$  ($i\in[s]$) and column totals $r$ ($j\in [M]$). See, for example,~\cite{DiaGan95}. Now each column $r_{ij}$ ($j\in [M]$)
can be interpreted as an $r$-multiset $\set{i_{1j},\ldots,i_{rj}}\subseteq [s]$, where $i\in[s]$ appears with multiplicity $r_{ij}$. Thus, choosing any $(\alpha_1,\ldots,\alpha_r)\in S$, we have
\begin{equation*}
  \prod_{i=1}^s\lambda_i^{M_i}\ =\ \prod_{j=1}^M \prod_{i=1}^s \lambda_i^{r_{ij}}\ =\ \prod_{j=1}^M \big( \lambda_{i_{1j}}\cdots\lambda_{i_{rj}}\big)\ =\
  \prod_{j=1}^M g( (\alpha_1,i_{1j}),\ldots,(\alpha_r,i_{rj}) ),
\end{equation*}
using \eqref{eq:40}. This can be computed in $O(M)$ time in $\rats(\theta)$, so $Z^g(G)$ can be evaluated in $O(M^s)$ time. The most demanding part of the  computation seems to be the $O(nM^s)$ time needed to determine the relevant monomials by dynamic programming. But clearly all computations can be done in \fp, and by working entirely within $\rats(\theta)$.

\end{document}